\documentclass[11pt,a4paper]{article}

\usepackage[T1]{fontenc}
\usepackage[utf8]{inputenc}
\usepackage{lmodern}
\usepackage[margin=1in]{geometry}
\usepackage{cite}
\usepackage{xcolor}
\usepackage{amsmath,amssymb,amsfonts,amsthm}
\usepackage{graphicx}
\usepackage{bm}
\usepackage{textcomp}
\usepackage{microtype}
\usepackage[hidelinks]{hyperref}
\hypersetup{pdftitle={Exact PID and PI Gain Regions for Uncertain Non-Affine MIMO Systems},
  pdfauthor={Cheng Zhao}}
\allowdisplaybreaks[2]
\newtheorem{theorem}{Theorem}
\newtheorem{proposition}{Proposition}
\newtheorem{lemma}{Lemma}
\newtheorem{corollary}{Corollary}
\newtheorem{remark}{Remark}
\theoremstyle{definition}
\newtheorem{definition}{Definition}

\newcommand{\R}{\mathbb{R}}

\newcommand{\Sym}{\operatorname{Sym}}
\newcommand{\norm}[1]{\lVert #1\rVert}
\renewcommand{\Re}{\operatorname{Re}}
\renewcommand{\Im}{\operatorname{Im}}
\newcommand{\Realify}{\mathcal{R}}

\begin{document}

\title{Exact PID and PI Gain Regions for Uncertain Non-Affine MIMO Systems}
\author{Cheng Zhao\thanks{State Key Laboratory of Mathematical Sciences,
Academy of Mathematics and Systems Science, Chinese Academy of Sciences,
Beijing 100190, China. E-mail: \href{mailto:zhaocheng@amss.ac.cn}{zhaocheng@amss.ac.cn}.}}
\date{}

\maketitle

\begin{abstract}
Characterizing all proportional--integral--derivative (PID) gains that guarantee uniform exponential
regulation of uncertain nonlinear systems remains a fundamental problem.
Most existing results provide only sufficient conditions derived from a
particular Lyapunov construction. This paper gives exact gain regions,
defined directly by uniform regulation over the entire uncertainty class,
for two classes of nonaffine multi-input multi-output (MIMO) systems
with $m\ge2$ controlled coordinates
and input channels. For systems of
second order under PID control, we independently characterize a common
quadratic inner region and an outer region based on a linear subclass.
The two regions coincide and therefore characterize the PID gain region
exactly at the system level. For systems of first order under PI control,
a lossless reduction to a linear family on the uncertainty boundary,
combined with a common quadratic certificate, yields the exact PI region.
Both regions strictly enlarge existing sufficient regions, and the PID
result also provides a sequential tuning rule. A compact linear family
with one uncertain parameter shows that common quadratic certification
is not generally lossless. These results establish endpoint coincidence
and boundary family equivalence as two routes to exact gain regions and
motivate the broader question of which nonlinear uncertainty classes
admit such characterizations.
\end{abstract}


\section{Introduction}

Proportional--integral--derivative (PID) control and its common PI and PD
variants remain among the most widely used feedback strategies in industrial
automation, robotics, and aerospace applications, owing to their relative
simplicity and ease of implementation~\cite{AstromHagglund2001}. Since closed-loop stability is a prerequisite for achieving a wide range of performance objectives, the widespread use of PID control makes reliable gain selection a problem of both practical and theoretical importance. A complete
stabilizing gain region does more than certify individual controllers: it
reveals the robustness limits of a prescribed plant class and provides a
stability-preserving domain for performance-oriented tuning \cite{ZhuZhaoGuo2026,leiPlenary2026}. For example,
an $H_\infty$ (worst-case frequency-response gain) criterion may be optimized over such a region \cite{BlanchiniEtAl2004}. Classical
extremum-seeking methods for feedback-gain tuning are typically analyzed
under stability of the frozen closed loops throughout the search domain,
whereas stability-certified policy-learning methods require the entire
time-varying policy sequence to remain certified~\cite{AriyurKrstic2003,
KillingsworthKrstic2006,PangBianJiang2022,TanNesicMareels2006}. A sharp stabilizing gain region, especially an exact one, therefore provides
a basis for defining a broad certified search domain, helping
preserve closed-loop stability throughout optimization without
excluding potentially high-performing controllers.

A mature theory for characterizing complete stabilizing PID regions has
developed primarily for linear plants, including fixed systems, interval
families, and time-delay models~\cite{HoDattaBhattacharyya1997,
DattaHoBhattacharyya2000,SilvaDattaBhattacharyya2005,OuZhangYu2009}.
Closely related gain-region questions and fundamental performance
limitations have also been studied for uncertain stochastic nonlinear
systems under PD control \cite{ZhaoZhang2024}. The corresponding linear PID theory relies primarily on polynomial,
quasipolynomial, frequency-domain, and parameter-space methods
\cite{ShafieiShenton1997,HuangWang2000,Hohenbichler2009,
LiNiculescuChenChai2021,HoDattaBhattacharyya2001,
MichielsNiculescu2007,AppeltansNiculescuMichiels2022}.
These tools can yield complete stabilizing regions or finite
robust-stability tests. Their reliance on a closed-loop characteristic
equation and its root distribution, however, prevents a direct extension to
general nonlinear uncertainty classes.

For uncertain nonlinear systems, substantial progress has been made through
Lyapunov, small-gain, differential-inclusion, and
Lyapunov--Krasovskii methods. These approaches have established global or
semiglobal regulation, exponential convergence, and constructive PID designs
for systems with structural uncertainty and extended PID feedback
\cite{ZhaoGuo2021}, stochastic perturbations \cite{ZhangGuo2020}, input or
time-varying delays \cite{LiMa2021,LyuLin2022}, and sampled-data
implementations \cite{ZhangFridman2022,LiMa2023}. Related
feedback-domination arguments have also yielded explicit sufficient gain
conditions for uncertain planar nonlinear systems
\cite{ChenWuHsu2026}.

Broadly speaking, existing theoretical approaches to PID design for
uncertain nonlinear systems are predominantly certificate-driven. They
typically select a particular Lyapunov function, comparison estimate,
small-gain argument, or linear matrix inequality (LMI) formulation and then derive sufficient gain
conditions. Such conditions are useful for controller design, but they
generally certify only the gains captured by the chosen construction.
Even when the Lyapunov matrices and auxiliary variables are optimized
over, the resulting gain set often remains implicit and tied to a
prescribed certificate class. Moreover, an exact solution of a
certificate-level feasibility problem need not recover the intrinsic
system-level gain region. The unresolved problem is therefore to
characterize all PID or PI gains that guarantee uniform exponential
regulation of every admissible nonlinear system, directly in terms of
the controller gains and uncertainty bounds and without restricting the
stability certificate a priori. This system-level all-gains problem is
particularly challenging for non-affine MIMO dynamics, for which neither
a common characteristic equation nor a direct scalar stability criterion
is available.

We address this problem through an inner--outer strategy. A common
quadratic certificate defines a tractable inner region, whereas a
carefully selected linear test family provides independently derived
necessary conditions and hence an outer region. The two regions are
characterized separately, after which system-level exactness is obtained
by proving that the gap between them closes. For the second-order PID
class, write $\Omega_{\rm pid}$ for the system-level gain region,
$\Omega_{\rm pid}^{\rm cq}$ for the common-quadratic inner region, and
$\Omega_{\rm pid}^{\rm lin}$ for the linear-subclass outer region.
This strategy then takes the form
\begin{equation*}
\Omega_{\rm pid}^{\rm cq}
\subseteq
\Omega_{\rm pid}
\subseteq
\Omega_{\rm pid}^{\rm lin}.
\end{equation*}
For the first-order PI class, the same principle is realized through a
lossless boundary-linear test family together with an explicit common
quadratic certificate. These two constructions provide complementary
routes to a complete gain-region characterization: inner--outer endpoint
coincidence for PID control and boundary-linear equivalence for PI
control. The analogous PI regions are denoted by $\Omega_{\rm pi}$ and
$\Omega_{\rm pi}^{\rm cq}$. Throughout, $\mathbb R_{>0}^r$ denotes the
set of real $r$-vectors with strictly positive components.

The contributions below make these two exactness mechanisms precise.
The main contributions are summarized as follows.

 \emph{i) Exact PID region through endpoint coincidence.}
For second-order non-affine MIMO systems, we characterize both
endpoints exactly and independently. The common-quadratic inner region is reduced to explicit algebraic
inequalities in the PID gains and two auxiliary scalars, whereas the
linear-subclass outer region is described by an exact frequency-domain
condition. Although
obtained through  different arguments, the two endpoints
coincide for every $m\geq 2$:
\begin{equation*}
\Omega_{\rm pid}^{\rm cq}
=
\Omega_{\rm pid}
=
\Omega_{\rm pid}^{\rm lin}.
\end{equation*}
Thus the complete system-level PID gain region is obtained:
common-quadratic certification loses no admissible gains, and the
linear subclass loses no necessary information. For $m=1$, both
endpoints are still characterized exactly, but they are generally
distinct and provide explicit inner and outer bounds on the
system-level region.

 \emph{ii) Exact PI region through boundary-linear equivalence.}
For the first-order non-affine MIMO class, we construct a countable
boundary-linear family and prove that, for $m\geq 2$, uniformly
exponentially regulating this family is equivalent to regulating the
full nonlinear uncertainty class. Combining this lossless boundary test
with an explicit common quadratic Lyapunov function yields
\begin{equation*}
\Omega_{\rm pi}^{\rm cq}
=
\Omega_{\rm pi}
=
\left\{
(k_p,k_i)\in\mathbb{R}_{>0}^{2}:
k_i+Lk_p<bk_p^2
\right\}.
\end{equation*}
This establishes system-level exactness through a mechanism
 different from the endpoint coincidence used for PID
control.

\emph{iii) Strict improvements and structural limits.}
The exact PID and PI regions strictly enlarge the corresponding
sufficient regions in~\cite{ZhaoGuo2022}. For PID control, the new
region captures high-gain scalings excluded by the previous condition
and yields a simple sequential tuning rule. We further construct a
compact one-parameter linear family whose system-level PID gain region
strictly exceeds its common-quadratic region. Thus common-quadratic
losslessness is a structural property of the uncertainty classes
considered here, rather than a generic feature of robust stability.

The remainder of the paper is organized as follows.
Section~\ref{sec:pid} develops the exact PID gain-region
characterization for second-order non-affine MIMO systems, including
the two endpoint regions, their coincidence for $m\geq2$, and the
resulting tuning implications. Section~\ref{sec:pi} establishes the
boundary-linear equivalence and the exact PI gain region for
first-order non-affine MIMO systems.
Section~\ref{sec:cq-losslessness-counterexample} presents a linear
counterexample showing that common-quadratic losslessness is not a
generic property. The proofs of the main PID and PI results are
collected in Section~\ref{sec:proof-main}, followed by the conclusion.

\paragraph{Notation.}
For a real or complex vector $v$, $\norm v=(v^*v)^{1/2}$ is its Euclidean
norm. For a matrix $A$, the same notation denotes the induced Euclidean
(spectral or $2$-) norm,
\[
\norm A:=\sup_{v\ne0}\frac{\norm{Av}}{\norm v}
=\sqrt{\lambda_{\max}(A^*A)}.
\]
The superscripts $\top$ and $*$ denote transpose and conjugate transpose,
respectively. For a real square matrix $A$, let
$\Sym(A):=(A+A^\top)/2$. The symbols $\mathbb S^r$, $\mathbb S_+^r$, and
$\mathbb S_{++}^r$ denote the real symmetric $r\times r$ matrices and their
positive-semidefinite and positive-definite cones. For symmetric matrices,
$A\succeq B$ (respectively, $A\succ B$) means that $A-B$ is positive
semidefinite (respectively, positive definite); $\preceq$ and $\prec$
have the reverse meanings. The symbols $\lambda_{\min}$ and
$\lambda_{\max}$ denote the smallest and largest eigenvalues of a symmetric
or Hermitian matrix. We write $I_r$ for the $r\times r$ identity,
$0_{p\times q}$ for the $p\times q$ zero matrix, and $0_r:=0_{r\times r}$;
unsubscripted zero blocks have the dimensions dictated by their positions.
The notation $\operatorname{col}(v_1,\ldots,v_j)$ means vertical stacking,
and $\operatorname{diag}(A_1,\ldots,A_j)$ denotes a block-diagonal matrix.
A matrix is \emph{Hurwitz} if all its eigenvalues have strictly negative
real parts; a polynomial is Hurwitz if all its zeros do. Finally,
$C^1(\mathbb R^p;\mathbb R^q)$ denotes continuously differentiable maps
from $\mathbb R^p$ to $\mathbb R^q$.

\section{PID control of second-order non-affine MIMO systems}
\label{sec:pid}

This section studies PID regulation for the second-order non-affine
MIMO uncertainty class and aims to \emph{determine all gains that guarantee
uniform exponential regulation} over the full nonlinear class. We first
formulate the system-level PID gain region directly in terms of
closed-loop trajectories. We then introduce and characterize two
independently derived endpoint regions: a common-quadratic inner region
and a linear-subclass outer region. After obtaining gain-only
descriptions of both endpoints, we prove that they coincide for
$m\geq2$, thereby yielding the exact system-level PID gain region. The
scalar case is considered separately. Finally, we compare the resulting
region with an existing sufficient region and derive a simple sequential
PID tuning rule.

\subsection{Problem formulation and the system-level PID gain region}
For an integer $m\ge1$, consider the following class of uncertain second-order non-affine MIMO
systems:
\begin{equation}
\label{eq:system}
\dot x_1=x_2,\qquad
\dot x_2=f(x_1,x_2,u),
\qquad x_1,x_2,u\in\R^m,
\end{equation}
where $x_1$ and $x_2$ denote generalized position and velocity,
respectively, $u$ is the control input, and
$f\in C^1(\R^{3m};\R^m)$ is an unknown nonlinear map that need not be affine
in $u$. The dimension $m$ is the number of controlled coordinates and input
channels. Subject to the bounds specified below, this class accommodates
coupled second-order agents and multi-degree-of-freedom force- or
torque-driven systems with state-dependent nonlinearities, cross-channel
coupling, and static input nonlinearities.

We apply the  PID controller with scalar gains
$k_p,k_i,k_d$:
\begin{align}
u(t)&=k_p e(t)+k_i\textstyle\int_0^t e(s)\,ds+k_d\dot e(t),
\label{eq:pid}\\
e(t)&=y^*-x_1(t).\nonumber
\end{align}
Here $y^*\in\mathbb{R}^m$ is a constant setpoint.

For prescribed constants $L_1,L_2,b>0$, let
$\mathcal F_{L_1,L_2,b}$ be the set of all functions
$f\in C^1(\R^{3m};\R^m)$ satisfying
\begin{equation}
\label{eq:function-class}
\begin{aligned}
\norm{\frac{\partial f}{\partial x_1}}\le L_1,\quad
\norm{\frac{\partial f}{\partial x_2}}\le L_2,\quad
\Sym\big(\frac{\partial f}{\partial u}\big)\succeq bI_m,
\end{aligned}
\end{equation}
where $\frac{\partial f}{\partial x_i}$, $i=1,2$, and
$\frac{\partial f}{\partial u}$ are the $m\times m$ Jacobian matrices
of $f$ with respect to $x_i$ and $u$, respectively. Writing
$f=(f_1,\ldots,f_m)^\top$, their entries are
\[
\left[\frac{\partial f}{\partial x_i}\right]_{rs}
=\frac{\partial f_r}{\partial(x_i)_s},
\qquad
\left[\frac{\partial f}{\partial u}\right]_{rs}
=\frac{\partial f_r}{\partial u_s},
\qquad r,s=1,\ldots,m.
\]
The matrix norms in \eqref{eq:function-class} are the induced $2$-norm
defined above, and all three bounds hold at every $(x_1,x_2,u)$. Strong
monotonicity and coercivity imply that (see Proposition 4.1 in \cite{ZhaoGuo2022}), for
each $f\in\mathcal F_{L_1,L_2,b}$ and $y^*\in\R^m$, there is a unique $u^*$
such that $f(y^*,0,u^*)=0$.

The objective is to select the PID gains using only the uncertainty bounds
$(L_1,L_2,b)$ so that every admissible closed loop is forward complete and
uniformly exponentially converges to the regulation equilibrium. This leads naturally to the following definition of the
system-level PID gain region, in which gain admissibility is determined
directly from the closed-loop trajectories of all systems in the uncertainty
class, rather than from a prescribed stability certificate.
Throughout, the dependence of the gain regions on the dimension $m$
and uncertainty bounds is suppressed for notational simplicity.

\begin{definition}[System-level robust PID gain region]
\label{def:pid-system-level}
A gain triple $(k_p,k_i,k_d)\in\mathbb R^3$ with $k_i\ne0$ belongs to
$\Omega_{\rm pid}$ if there exist constants $M\ge1$ and $\lambda>0$,
depending only on the gains and the uncertainty bounds $(L_1,L_2,b)$, such
that, for every $f\in\mathcal F_{L_1,L_2,b}$, every constant reference, and
every initial condition, the corresponding closed-loop solution is forward
complete and satisfies
\begin{equation}
\label{eq:full-exp-certificate}
\norm{z(t)}
\le
Me^{-\lambda t}\norm{z(0)},
\qquad t\ge0,
\end{equation}
where
\[
z=
\operatorname{col}\!\left(
\eta(t)-u^*/k_i,\,e(t),\,\dot e(t)
\right),
\qquad
\eta(t)=\int_0^t e(s)\,ds.
\]
The set $\Omega_{\rm pid}$ is called the
\emph{system-level robust PID gain region}.
\end{definition}

\begin{remark}\normalfont
The restriction $k_i\ne0$ in Definition~\ref{def:pid-system-level}
entails no loss of admissible PID gains. Indeed, the class
$\mathcal F_{L_1,L_2,b}$ constrains $f$ only through its partial derivatives,
so any constant additive disturbance can be absorbed into $f$ without
altering the uncertainty bounds. Robust asymptotic regulation over the full
uncertainty class therefore requires integral action. Hence every admissible
gain triple necessarily satisfies $k_i\ne0$, and the term $u^*/k_i$ is well
defined.
\end{remark}

Our main objective is to characterize $\Omega_{\rm pid}$ exactly, that is,
to derive necessary and sufficient algebraic conditions for membership in
terms of $(k_p,k_i,k_d)$ and the uncertainty bounds $(L_1,L_2,b)$.
Unlike a sufficient gain condition, an exact characterization identifies all
gains that guarantee uniform exponential regulation over the entire
uncertainty class. It therefore makes explicit the robustness
limits under the prescribed uncertainty bounds and provides the complete
admissible domain for performance-oriented tuning.

Obtaining such a characterization directly is difficult because the
admissible dynamics may depend nonlinearly and non-affinely on vector-valued
states and inputs, with coupling across channels. Consequently, neither a
common characteristic equation nor a direct scalar stability criterion is
available. Rather than attacking this nonlinear trajectory-level set
directly, we adopt an inner--outer strategy based on the inclusion chain
\[
\Omega_{\rm pid}^{\rm cq}
\subseteq
\Omega_{\rm pid}
\subseteq
\Omega_{\rm pid}^{\rm lin},
\]
where $\Omega_{\rm pid}^{\rm cq}$ is the common-quadratic inner endpoint and
$\Omega_{\rm pid}^{\rm lin}$ is the linear-subclass outer endpoint (to be defined later).

\subsection{Common-quadratic inner region}\label{sec:cq-pid-region}

We first ask a certificate-level question: which PID gains admit \emph{some} common quadratic Lyapunov function for the entire induced matrix family? Unlike a construction based on a prescribed quadratic form, this searches over the whole common-quadratic class.

It is worth noting that the nonemptiness of $\Omega_{\rm pid}$ follows from
~\cite{ZhaoGuo2022}, which constructed an explicit open and
unbounded sufficient region contained in $\Omega_{\rm pid}$ using a
\emph{particular common quadratic Lyapunov function}. Since their construction does not reveal whether alternative
certificates can admit additional gains, it motivates determining the
 set of PID gains that admit a common quadratic certificate. This
certificate-level problem provides a tractable first step toward an exact
characterization of the system-level robust PID gain region.

To formulate this problem, we identify the associated uncertain closed-loop
matrix family. Consider first a linear member of the uncertainty class,
written about the equilibrium $(y^*,0,u^*)$:
\[
f(x_1,x_2,u)
=
D_1(x_1-y^*)+D_2x_2+\Theta(u-u^*).
\]
Here $D_1,D_2,\Theta\in\mathbb R^{m\times m}$ are the constant state
and input Jacobians of this linear member.
For $k_i\ne0$, define
\(
z:=\operatorname{col}\!\left(
\eta-u^*/k_i,\,e,\,\dot e
\right).
\)
Using $x_1-y^*=-e$, $x_2=-\dot e$, and
\[
u-u^*
=
k_i\left(\eta-u^*/k_i\right)
+k_pe+k_d\dot e,
\]
the augmented error dynamics become
\begin{align}\label{linear-repre}
\dot z=\mathcal A(D_1,D_2,\Theta)z,
\end{align}
where
\[
\mathcal A(D_1,D_2,\Theta):=
\begin{bmatrix}
0&I_m&0\\
0&0&I_m\\
-k_i\Theta&D_1-k_p\Theta&D_2-k_d\Theta
\end{bmatrix}.
\]

For a general nonlinear function $f\in\mathcal F_{L_1,L_2,b}$, the mean-value theorem applied
along the line segment joining $(y^*,0,u^*)$ to $(x_1,x_2,u)$ gives the same
representation \eqref{linear-repre} with \emph{state-dependent matrices} satisfying
\[
\norm{D_j(z)}\le L_j,\quad j=1,2,
\qquad
\Sym(\Theta(z))\succeq bI_m.
\]

Thus every admissible nonlinear closed loop is embedded pointwise in the
uncertain matrix family $\mathcal A(D_1,D_2,\Theta)$. This representation
leads to the following definition.

\begin{definition}[Common-quadratic PID gain region]
\label{def:cq-pid-region}
A gain triple $(k_p,k_i,k_d)\in\mathbb R^3$ is called
\emph{common-quadratically admissible} if there exist a common
$P=P^\top\succ0$ in $\mathbb R^{3m\times3m}$ and $\alpha>0$ such that
\begin{equation}
\label{eq:robust-matrix-inequality}
P\mathcal A(D_1,D_2,\Theta)
+\mathcal A(D_1,D_2,\Theta)^\top P
\preceq-\alpha I_{3m}
\end{equation}
for every $D_1,D_2,\Theta\in\mathbb R^{m\times m}$ satisfying
\begin{equation}
\label{eq:admissible-matrix-triples}
\norm{D_1}\le L_1,\qquad
\norm{D_2}\le L_2,\qquad
\Sym(\Theta)\succeq bI_m.
\end{equation}
Any such $P$ is called a \emph{common quadratic certificate}. The set of
all common-quadratically admissible PID gains is denoted by
$\Omega_{\rm pid}^{\rm cq}$ and called the
\emph{common-quadratic PID gain region}.
\end{definition}

The following proposition connects this certificate-level notion to the
system-level PID region by showing that every common-quadratically
admissible PID gain belongs to $\Omega_{\rm pid}$.

\begin{proposition}\normalfont
\label{prop:cq-certificate}
Every common-quadratically admissible PID gain triple guarantees the uniform
exponential regulation \eqref{eq:full-exp-certificate} over the
uncertainty class $\mathcal F_{L_1,L_2,b}$. Hence,
\(\Omega_{\rm pid}^{\rm cq}
\subseteq
\Omega_{\rm pid}.
\)
\end{proposition}

\begin{proof}
Let $P=P^\top\succ0$ and $\alpha>0$ be a common quadratic certificate and
its uniform decay margin. Taking $D_1=D_2=0$ and $\Theta=bI_m$ in
\eqref{eq:robust-matrix-inequality} shows that
$\mathcal A(0,0,bI_m)$ is Hurwitz. Its characteristic polynomial in each
channel is
\[
s^3+bk_ds^2+bk_ps+bk_i.
\]
The Routh--Hurwitz criterion therefore gives
\[
k_i>0,\qquad k_p>0,\qquad k_d>0,
\qquad bk_pk_d>k_i.
\]
In particular, the shifted integral state in
\eqref{eq:full-exp-certificate} is well defined.

Fix any admissible nonlinearity, constant reference, and plant initial
condition. Write
\[
z=\operatorname{col}(z_0,z_1,z_2),
\qquad
q:=k_iz_0+k_pz_1+k_dz_2.
\]
The segmentwise mean-value decomposition described above, and used similarly
in the proof of Theorem~3.1 in~\cite{ZhaoGuo2022}, gives
\begin{equation}
\label{eq:pid-mean-value}
f(y^*-z_1,-z_2,u^*+q)
=
-D_1(z)z_1-D_2(z)z_2+\Theta(z)q,
\end{equation}
where
\[
\norm{D_j(z)}\le L_j,\quad j=1,2,
\qquad
\Sym(\Theta(z))\succeq bI_m.
\]
Consequently, the augmented error dynamics satisfy
\begin{equation}
\label{eq:certificate-dynamics}
\dot z
=
\mathcal A\bigl(D_1(z),D_2(z),\Theta(z)\bigr)z,
\end{equation}
with an admissible matrix triple at every closed-loop state.

For $V(z):=z^\top Pz$, the common quadratic inequality yields
\[
\dot V
\le
-\alpha\norm{z}^2
\le
-\frac{\alpha}{\lambda_{\max}(P)}V.
\]
It follows that $\norm{z(t)}
\le
M e^{-\lambda t}\norm{z(0)},$  where
\[
M:=
\sqrt{\frac{\lambda_{\max}(P)}{\lambda_{\min}(P)}},
\qquad
\lambda:=
\frac{\alpha}{2\lambda_{\max}(P)}.
\]
Because $P$ and $\alpha$ are common to the entire matrix family, $M$ and
$\lambda$ are independent of the admissible nonlinearity, reference, and
initial condition. This proves \eqref{eq:full-exp-certificate}.

Finally, $f\in C^1$ makes the closed-loop vector field locally Lipschitz,
while the preceding estimate prevents finite escape. The standard
continuation theorem therefore gives forward completeness.
\end{proof}

Proposition~\ref{prop:cq-certificate} rigorously establishes the inclusion
$\Omega_{\rm pid}^{\rm cq}\subseteq\Omega_{\rm pid}$, thereby identifying
$\Omega_{\rm pid}^{\rm cq}$ as a common-quadratic inner region of the
system-level PID gain region. Its membership condition, however, remains
implicit through a semi-infinite LMI. We next characterize this inner region
exactly. For $m=1$, the semi-infinite condition reduces losslessly to a
four-vertex LMI test, whereas for $m\ge2$, the Lyapunov matrix can be
eliminated entirely, yielding an explicit gain-only characterization in
terms of two auxiliary scalars.

For $m=1$, let $\mathcal V:=\{-1,1\}$ denote the vertex-sign set and define,
for $s,t\in\mathcal V$,
\begin{equation}
\label{eq:scalar-vertices}
\mathcal A_{st}^{\rm sc}:=
\begin{bmatrix}
0&1&0\\
0&0&1\\
-bk_i&sL_1-bk_p&tL_2-bk_d
\end{bmatrix}.
\end{equation}
Set $e_3=\begin{bmatrix}0&0&1\end{bmatrix}^{\!\top}$,
$\kappa:=(k_i,k_p,k_d)^\top$. The associated \emph{four-vertex feasibility
condition} is
\begin{equation}
\label{eq:four-vertex-condition}
\begin{aligned}
&\exists\,P_0=P_0^\top\succ0,\qquad P_0e_3=\kappa,\\
&P_0\mathcal A_{st}^{\rm sc}
+(\mathcal A_{st}^{\rm sc})^\top P_0\prec0,
\quad (s,t)\in\mathcal V^2.
\end{aligned}
\end{equation}

For the case $m\geq 2$, let $\mu,\rho\in(0,1)$ and define
\begin{align}
\label{eq:Pi}
\Pi_{\mu,\rho}
&:=k_p^2-
\frac{L_1^2}{4b^2\mu\rho(1-\mu)},\\
\label{eq:Delta}
\Delta_{\mu,\rho}
&:=k_d^2-
\frac{k_p}{b(1-\mu)}
-\frac{L_2^2}
{4b^2\mu(1-\rho)(1-\mu)}.
\end{align}
When $(\mu,\rho)$ are fixed and no confusion can arise, we write
$\Pi:=\Pi_{\mu,\rho}$ and $\Delta:=\Delta_{\mu,\rho}$ for brevity.

The next theorem gives the exact common-quadratic region in both the scalar
and multivariable cases.

\begin{theorem}[Exact common-quadratic PID gain region]
\label{thm:main}
The exact common-quadratic PID gain regions are as follows:

i) For scalar systems $m=1$, a gain triple belongs to $\Omega_{\rm pid}^{\rm cq}$
if and only if it satisfies the four-vertex condition
\eqref{eq:four-vertex-condition}.

ii) For every $m\ge2$,
\begin{equation}
\label{eq:exact-cq-region}
\Omega_{\rm pid}^{\rm cq}=\bigcup_{\mu,\rho\in(0,1)}
\left\{
(k_p,k_i,k_d)\in\mathbb R_{>0}^3
\ \middle|\
\Pi_{\mu,\rho}>0,\quad \Delta_{\mu,\rho}>0,\quad
k_i<
\frac{\Pi_{\mu,\rho}}
{2\bigl(k_d-\sqrt{\Delta_{\mu,\rho}}\bigr)}
\right\}.
\end{equation}
\end{theorem}

It is worth mentioning that for every gain in \eqref{eq:exact-cq-region}, one has
$0<\Delta_{\mu,\rho}<k_d^2$ for the corresponding auxiliary parameters,
so the denominator is positive. Moreover, every gain in the multivariable
region \eqref{eq:exact-cq-region} also satisfies the scalar four-vertex
condition.
The proof is deferred to Section~\ref{sec:proof-main}.

\begin{remark}\normalfont
Theorem \ref{thm:main} is a lossless elimination result at the certificate level: all matrix-valued certificate variables and uncertainty matrices have been removed. It does \emph{not}, by itself, imply that $\Omega_{\rm pid}^{\rm cq}$
 exhausts the system-level region $\Omega_{\rm pid}$; this requires the independent outer analysis below.

%Let $\mathcal U$ denote the set of matrix triples satisfying
%\eqref{eq:admissible-matrix-triples}. Theorem~\ref{thm:main}  shows that, for $m\geq 2$
%\begin{equation*}
%\begin{aligned}
%&\exists \,P=P^\top\succ0,\ \alpha>0\ \text{such that for all }\xi\in\mathcal U,\\[-1mm]
%&P\mathcal A(D_1,D_2,\Theta)+\mathcal A(D_1,D_2,\Theta)^\top P
%\preceq-\alpha I_{3m}
%\\
%\Longleftrightarrow&\quad
%(k_p,k_i,k_d) \text{ belongs to the set defined in } \eqref{eq:exact-cq-region}.
%\end{aligned}
%\end{equation*}
%Thus the final characterization involves only the PID gains and the uncertainty
%bounds $(L_1,L_2,b)$, with neither the
%matrix $P$ nor the uncertainty matrices $(D_1,D_2,\Theta)$ remaining.
\end{remark}

\subsection{Linear-subclass outer region}

The preceding subsection gives an exact \emph{inner bound} on the
system-level region $\Omega_{\rm pid}$. To obtain a complementary outer bound, we now restrict
the full nonlinear uncertainty class to its admissible linear subclass. Any
gain achieving system-level robust regulation must, in particular, stabilize
every member of this subclass.

\begin{definition}[Linear-subclass PID gain region]
\label{def:pid-linear-region}
Let $\mathcal F_m^{\mathrm{lin}}$ denote the family of linear maps
\[
f_{D_1,D_2,\Theta}(x_1,x_2,u)
:=
D_1x_1+D_2x_2+\Theta u
\]
generated by matrices satisfying
\(
\norm{D_1}\le L_1,\
\norm{D_2}\le L_2,\
\Sym(\Theta)\succeq bI_m.
\)
The linear-subclass PID gain region is defined by
\[
\begin{aligned}
\Omega_{\rm pid}^{\rm lin}
:=
\Bigl\{(k_p,k_i,k_d)\in\mathbb R^3:\;
&\ \mathcal A(D_1,D_2,\Theta)\text{ is Hurwitz for every }
f_{D_1,D_2,\Theta}\in\mathcal F_m^{\mathrm{lin}}
\Bigr\}.
\end{aligned}
\]
\end{definition}

Since $\mathcal F_m^{\rm lin}\subseteq\mathcal F_{L_1,L_2,b}$,
every system-level admissible gain must regulate each member of the
linear subclass. More precisely, take $y^*=0$, so that $u^*=0$, and set
\[
\mathcal S_{\rm pid}:=\{0\}\times\mathbb R^m\times\mathbb R^m.
\]
For each fixed $\mathcal A=\mathcal A(D_1,D_2,\Theta)$, the system-level
estimate holds on $\mathcal S_{\rm pid}$. Since
\[
\mathcal S_{\rm pid}+\mathcal A\mathcal S_{\rm pid}
=\mathbb R^{3m},
\qquad
e^{\mathcal A t}\mathcal A=\mathcal A e^{\mathcal A t},
\]
exponential decay extends to the whole augmented state space, and hence
$\mathcal A$ is Hurwitz. Therefore,
$\Omega_{\rm pid}\subseteq\Omega_{\rm pid}^{\rm lin}$.

 We next characterize this outer approximation exactly in terms of the
PID gains and the uncertainty bounds.

\begin{proposition}[Exact linear-subclass PID gain region]
\label{prop:pid-linear-outer}
For $m=1$, denote \(\
\bar k_0=bk_i,\
\bar k_1=bk_p-L_1\) and \(\ \bar k_2=bk_d-L_2,
\) then
\begin{equation}
\label{eq:pid-scalar-linear-region}
\Omega_{\rm pid}^{\rm lin}
=
\left\{
(k_p,k_i,k_d)\in\mathbb R_{>0}^3:
\bar k_1>0,\
\bar k_1\bar k_2>\bar k_0
\right\}.
\end{equation}
For $m\ge2$,
\begin{equation}
\label{eq:pid-linear-region}
\Omega_{\rm pid}^{\rm lin}
=
\left\{
(k_p,k_i,k_d)\in\mathbb R_{>0}^3
\;:\;
\begin{aligned}
b\chi(\omega)-k_p\omega^4
>
(L_1\omega+L_2\omega^2)\sqrt{\chi(\omega)},
\quad \forall\,\omega>0
\end{aligned}
\right\}
\end{equation}
where
\begin{equation}
\label{eq:pid-frequency-function}
\chi(\omega)
:=
(k_i-k_d\omega^2)^2+k_p^2\omega^2.
\end{equation}
\end{proposition}

The proof is deferred to Section~\ref{sec:proof-main}.

\subsection{Endpoint coincidence and the exact PID gain region}

For $m\ge2$, the two endpoint regions are obtained through independent
arguments and admit different exact descriptions. The
common-quadratic inner region $\Omega_{\rm pid}^{\rm cq}$ is characterized
by algebraic inequalities involving two auxiliary parameters, whereas the
linear-subclass outer region $\Omega_{\rm pid}^{\rm lin}$ is described by a
frequency-domain inequality over all positive frequencies. Their coincidence
is therefore far from evident a priori.

Nevertheless, the general inclusion chain
\[
\Omega_{\rm pid}^{\rm cq}
\subseteq
\Omega_{\rm pid}
\subseteq
\Omega_{\rm pid}^{\rm lin}
\]
shows that equality of the two endpoints would close the inner--outer gap
and determine the system-level region exactly. The following
theorem shows that this is indeed the case.

\begin{theorem}[Endpoint coincidence and exact PID gain region]
\label{thm:pid-region-coincidence}
For every $m\ge2$, the three gain regions coincide:
\begin{equation}
\label{eq:pid-all-regions-equal}
\Omega_{\rm pid}^{\rm cq}
=
\Omega_{\rm pid}
=
\Omega_{\rm pid}^{\rm lin}.
\end{equation}
\end{theorem}

The proof is deferred to
Section~\ref{sec:proof-main}.

\begin{remark}\normalfont
Theorem~\ref{thm:pid-region-coincidence}, together with the exact
characterizations of the two endpoint regions, is the \emph{central result}
of this paper and has three complementary consequences. First, it provides a
complete gain-only characterization of the system-level PID region
$\Omega_{\rm pid}$ in terms of the PID gains and the uncertainty bounds
$(L_1,L_2,b)$. Second, every gain that uniformly exponentially regulates the
full nonlinear uncertainty class admits a common quadratic certificate;
hence common-quadratic certification is lossless at the gain-region level.
Third, stabilizing every admissible linear closed-loop matrix is sufficient
for robust regulation of the full nonlinear class; hence the linear
subclass is a lossless outer test.
\end{remark}

\subsection{The endpoint gap in the scalar case}
The coincidence above is a genuinely multivariable phenomenon. To clarify
the scope of Theorem~\ref{thm:pid-region-coincidence}, we now turn to the
scalar case $m=1$. Both endpoint regions remain exactly characterized:
$\Omega_{\rm pid}^{\rm cq}$ by the four-vertex LMIs in
\eqref{eq:four-vertex-condition}, and
$\Omega_{\rm pid}^{\rm lin}$ by
\eqref{eq:pid-scalar-linear-region}. However, the inner--outer gap no longer
closes, as shown next.

\begin{proposition}\normalfont
\label{prop:pid-scalar-strict}
For $m=1$,
\begin{equation}
\label{eq:pid-scalar-strict-inclusion}
\Omega_{\rm pid}^{\rm cq}
\subsetneq
\Omega_{\rm pid}^{\rm lin}.
\end{equation}
\end{proposition}

\begin{proof}
The inclusion
$\Omega_{\rm pid}^{\rm cq}\subseteq\Omega_{\rm pid}^{\rm lin}$
follows directly from the definitions, since a common quadratic certificate
makes every admissible linear closed-loop matrix Hurwitz.

To prove strictness, suppose that
$(k_p,k_i,k_d)\in\Omega_{\rm pid}^{\rm cq}$. By
Theorem~\ref{thm:main}, there exists $P_0=P_0^\top\succ0$ satisfying
\eqref{eq:four-vertex-condition}. In particular, with
\[
Q_{++}:=
-P_0\mathcal A_{++}^{\rm sc}
-(\mathcal A_{++}^{\rm sc})^\top P_0\succ0,
\]
the normalization $P_0e_3=(k_i,k_p,k_d)^\top$ gives
\begin{align}
0
&<e_3^\top Q_{++}e_3 \notag\\
&=-2(P_0e_3)^\top\mathcal A_{++}^{\rm sc}e_3
 =2\bigl[k_d(bk_d-L_2)-k_p\bigr].
\label{eq:scalar-cq-necessary}
\end{align}
Hence every gain in $\Omega_{\rm pid}^{\rm cq}$ necessarily satisfies
\[
k_d(bk_d-L_2)>k_p.
\]

We now construct a gain in $\Omega_{\rm pid}^{\rm lin}$ that violates this
condition. Choose $\varepsilon_p,\varepsilon_d>0$ such that
\(
\varepsilon_d(L_2+\varepsilon_d)
<
L_1+\varepsilon_p,
\)
and set
\[
k_p:=\frac{L_1+\varepsilon_p}{b},
\qquad
k_d:=\frac{L_2+\varepsilon_d}{b},
\qquad
k_i:=\frac{\varepsilon_p\varepsilon_d}{2b}.
\]
Then
\(
\bar k_1=\varepsilon_p,\
\bar k_2=\varepsilon_d,\
\bar k_0=\frac{\varepsilon_p\varepsilon_d}{2},
\)
so that
\[
\bar k_1>0,
\qquad
\bar k_1\bar k_2>\bar k_0.
\]
Proposition~\ref{prop:pid-linear-outer} therefore gives
$(k_p,k_i,k_d)\in\Omega_{\rm pid}^{\rm lin}$. On the other hand,
\[
bk_d(bk_d-L_2)
=\varepsilon_d(L_2+\varepsilon_d)
<
L_1+\varepsilon_p
=bk_p,
\]
which contradicts the necessary condition
\eqref{eq:scalar-cq-necessary}. Thus this gain does not belong to
$\Omega_{\rm pid}^{\rm cq}$, proving the proper inclusion.
\end{proof}

\begin{remark}\normalfont
\label{rem:pid-region-relations}
Proposition~\ref{prop:pid-scalar-strict}, together with the general
inner--outer inclusions, gives
\[
\Omega_{\rm pid}^{\rm cq}
\subseteq
\Omega_{\rm pid}
\subseteq
\Omega_{\rm pid}^{\rm lin},
\qquad m=1.
\]
Thus the two exact endpoints bracket, but do not determine, the scalar system-level PID region.
\end{remark}

\subsection{Strict improvement and sequential PID tuning}

We now compare the exact system-level region obtained above with the
sufficient region derived in Zhao and Guo~\cite{ZhaoGuo2022}. Their region, denoted
by $\Omega_{\rm pid}^{\rm ZG}$, consists of all
$(k_p,k_i,k_d)\in\mathbb R_{>0}^3$ satisfying
\begin{equation}
\label{eq:ZG}
k_p^2>2k_i k_d+\bar k,
\quad
k_d^2>k_p/b+\bar k,
\end{equation}
where
\(
\bar k:=(L_1+L_2)(k_p+k_d)/b.
\)

\begin{proposition}\normalfont
\label{prop:strict}
For every $m\ge2$,
\begin{equation}
\label{eq:strict-inclusion}
\Omega_{\rm pid}^{\rm ZG}
\subsetneq
\Omega_{\rm pid}.
\end{equation}
Moreover, for every $c\in(0,b)$, there exists $T_c>0$ such that
\begin{equation}
\label{eq:high-gain-family}
(t,ct^2,t)
\in
\Omega_{\rm pid}\setminus\Omega_{\rm pid}^{\rm ZG},
\qquad t\ge T_c.
\end{equation}
\end{proposition}

\begin{proof}
By Theorem~\ref{thm:pid-region-coincidence}, it suffices to prove
\(
\Omega_{\rm pid}^{\rm ZG}
\subsetneq
\Omega_{\rm pid}^{\rm cq}
\)
and establish the corresponding high-gain claim for
$\Omega_{\rm pid}^{\rm cq}$.

\emph{Inclusion.}
For every gain triple in $\Omega_{\rm pid}^{\rm ZG}$,
Zhao and Guo~\cite[Prop.~4.2]{ZhaoGuo2022} construct the common quadratic
certificate
\begin{equation*}
P_{\rm ZG}^{(m)}:=
\begin{bmatrix}
2k_ik_pbI_m&2k_ik_dbI_m&k_iI_m\\
2k_ik_dbI_m&(2k_pk_db-k_i)I_m&k_pI_m\\
k_iI_m&k_pI_m&k_dI_m
\end{bmatrix}.
\end{equation*}
Hence,
\(
\Omega_{\rm pid}^{\rm ZG}
\subseteq
\Omega_{\rm pid}^{\rm cq}.
\)

\emph{Strictness.}
Fix $c\in(0,b)$, choose $0<\mu<1-c/b$ and $\rho\in(0,1)$, and set
\begin{equation*}
\begin{aligned}
a_\mu&:=\frac{1}{b(1-\mu)},\qquad
C_1:=\frac{L_1^2}{4b^2\mu\rho(1-\mu)},\\
C_2&:=\frac{L_2^2}{4b^2\mu(1-\rho)(1-\mu)}.
\end{aligned}
\end{equation*}
For $(k_p,k_i,k_d)=(t,ct^2,t)$, define
\[
\Pi(t):=t^2-C_1,
\qquad
\Delta(t):=t^2-a_\mu t-C_2.
\]
Both quantities are positive for all sufficiently large $t$. The admissible
upper bound on $k_i$ in \eqref{eq:exact-cq-region} is
\[
K_i(t):=
\frac{\Pi(t)}
{2\bigl(t-\sqrt{\Delta(t)}\bigr)}.
\]
Rationalization gives
\begin{align*}
2\bigl(t-\sqrt{\Delta(t)}\bigr)
&=\frac{2(a_\mu t+C_2)}
{t+\sqrt{\Delta(t)}}\longrightarrow a_\mu,\\
\frac{K_i(t)}{t^2}
&\longrightarrow\frac{1}{a_\mu}
=b(1-\mu)>c.
\end{align*}
Therefore, $ct^2<K_i(t)$ for all sufficiently large $t$, and hence
\[
(t,ct^2,t)\in\Omega_{\rm pid}^{\rm cq}.
\]

For the region $\Omega_{\rm pid}^{\rm ZG}$,
$\bar k=2(L_1+L_2)t/b$, and its first defining inequality would require
\(
t^2>2ct^3+\frac{2(L_1+L_2)t}{b},
\)
which fails for all sufficiently large $t$. Thus,
\[
(t,ct^2,t)
\in
\Omega_{\rm pid}^{\rm cq}\setminus\Omega_{\rm pid}^{\rm ZG}
\]
for all sufficiently large $t$. This proves
$\Omega_{\rm pid}^{\rm ZG}\subsetneq\Omega_{\rm pid}^{\rm cq}$ and the
high-gain claim at the certificate level. The equality
$\Omega_{\rm pid}^{\rm cq}=\Omega_{\rm pid}$ from
Theorem~\ref{thm:pid-region-coincidence} then yields
\eqref{eq:strict-inclusion} and \eqref{eq:high-gain-family}.
\end{proof}

\begin{remark}\normalfont
The threshold $c<b$ in \eqref{eq:high-gain-family} is sharp: along this family, it coincides with the Routh--Hurwitz condition for the admissible nominal system corresponding to $D_1=D_2=0$ and $\Theta=bI_m$.
\end{remark}

The strict enlargement in \eqref{eq:strict-inclusion} is illustrated in Figure~\ref{fig:pid-region-comparison} on the $(k_p,k_d)$ cross-section obtained by fixing $k_i=0.5$. For $L_1=L_2=5$ and $b=1$, the exact region $\Omega_{\rm pid}$ extends well beyond the existing sufficient region $\Omega_{\rm pid}^{\rm ZG}$, admitting a substantially broader range of proportional and derivative gains.
\begin{figure}[t]
\centering
\includegraphics[width=0.72\linewidth]
{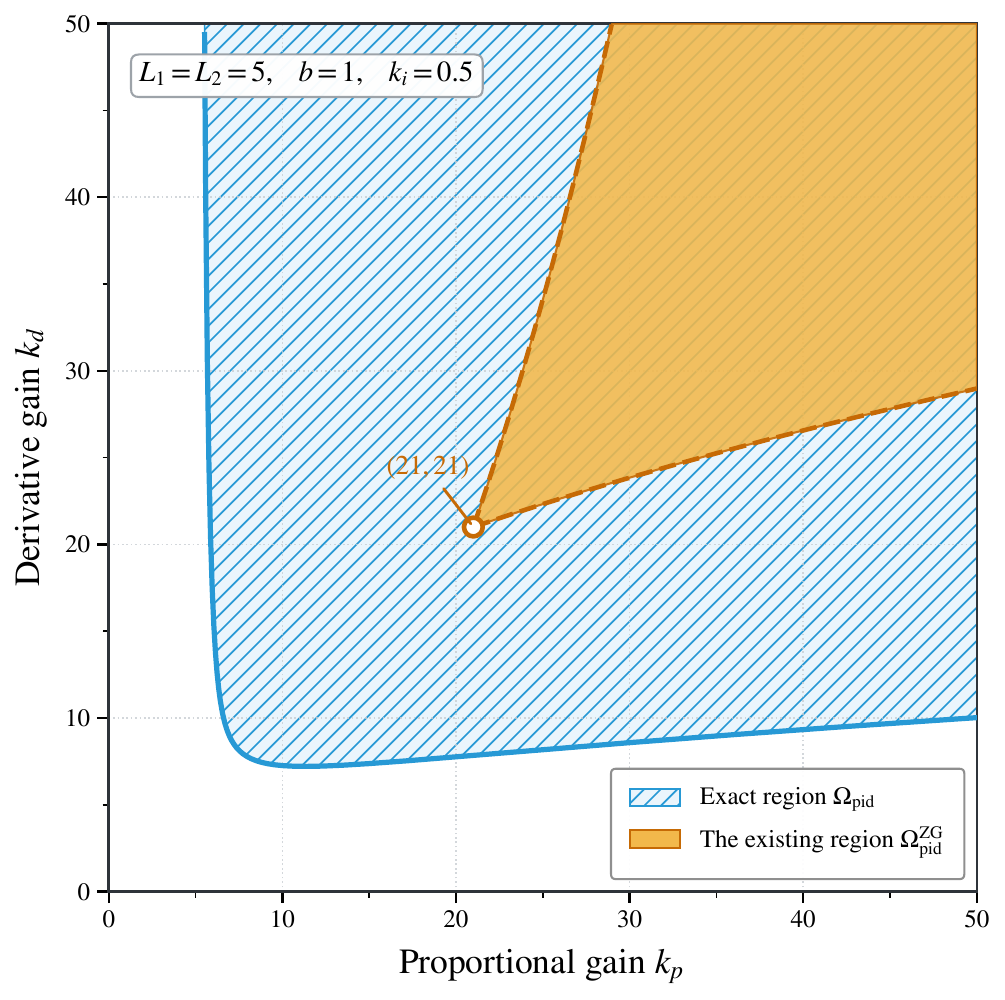}
\caption{Comparison of the PID gain regions for $L_1=L_2=5$, $b=1$, and $k_i=0.5$. The orange area represents the sufficient region $\Omega_{\rm pid}^{\rm ZG}$ derived in \cite{ZhaoGuo2022}, whereas the blue hatched area represents the exact region $\Omega_{\rm pid}$ established in this paper. The open circle at $(21,21)$ marks the intersection of the two boundary curves of the existing region.}
\label{fig:pid-region-comparison}
\end{figure}

The exact description of $\Omega_{\rm pid}^{\rm cq}$ involves a two-dimensional
search over $(\mu,\rho)$.  The balanced choice $\mu=\rho=1/2$ sacrifices some
flexibility in exchange for a transparent sequential tuning rule.

\begin{corollary}[A Simple Sequential PID Tuning Rule]
\label{cor:balanced}
 The PID gains can
be selected sequentially as follows:
\begin{enumerate}
\renewcommand{\labelenumi}{\roman{enumi})}

\item Choose
\[
k_p>\frac{\sqrt{2}\,L_1}{b}.
\]

\item Having fixed $k_p$, choose
\[
k_d>
\sqrt{\frac{2k_p}{b}+\frac{2L_2^2}{b^2}}.
\]

\item Define
\(
\Delta_*:=
k_d^2-\frac{2k_p}{b}-\frac{2L_2^2}{b^2},
\)
and choose
\[
0<k_i<
\frac{k_p^2-2L_1^2/b^2}
{2\bigl(k_d-\sqrt{\Delta_*}\bigr)}.
\]
\end{enumerate}
Every gain triple obtained by this procedure belongs to
$\Omega_{\rm pid}^{\rm cq}$ for every $m\ge2$.
\end{corollary}

\begin{proof}
Set $\mu=\rho=1/2$ in
\eqref{eq:Pi}, \eqref{eq:Delta}, and \eqref{eq:exact-cq-region}. Then
\begin{equation*}
\Pi_{1/2,1/2}=k_p^2-\frac{2L_1^2}{b^2},\qquad
\Delta_{1/2,1/2}=\Delta_*.
\end{equation*}
The three conditions  are exactly the resulting
positivity and integral-gain inequalities.
\end{proof}

\section{PI control of first-order non-affine MIMO systems}
\label{sec:pi}
The PID result above achieves system-level exactness by closing an
inner--outer gap. The first-order PI class exhibits an analogous
exactness phenomenon, but through a different mechanism: necessity is
captured by a boundary-linear test family, whereas sufficiency follows
from a common quadratic certificate. We first define the corresponding
system-level and common-quadratic PI regions, then establish the
boundary-linear equivalence and the exact gain condition, and finally
compare the result with the existing sufficient region and discuss the
scalar case.

\subsection{Problem formulation and the system-level PI gain region}

Let $m\ge1$ and consider the first-order non-affine system
\begin{equation}\label{sys:first-order}
\dot x(t)=f(x(t),u(t)),\qquad x(t),u(t)\in\R^m,
\end{equation}
where, for prescribed constants $L\ge0$ and $b>0$,
$\mathcal F_{L,b}$ consists of all
$f\in C^1(\R^{2m};\R^m)$ satisfying, for every $(x,u)$,
\begin{equation*}
\begin{aligned}
\big\|\frac{\partial f}{\partial x}\big\|\le L,\qquad
\Sym\!\big(\frac{\partial f}{\partial u}\big)\succeq bI_m.
\end{aligned}
\end{equation*}
Here $\partial f/\partial x$ and $\partial f/\partial u$ are again
$m\times m$ Jacobian matrices, with entries
$\partial f_r/\partial x_s$ and $\partial f_r/\partial u_s$, respectively;
the norm is the induced $2$-norm specified in the notation paragraph.
For a constant reference $y^*\in\R^m$, apply the PI controller
\begin{equation}
\label{eq:pi-controller}
\begin{aligned}
u(t)=k_pe(t)+k_i\int_0^t e(s)\,ds,\qquad
e(t)=y^*-x(t).
\end{aligned}
\end{equation}
 Set $\eta(t):=\int_0^t e(s)\,ds$, so that
$\dot\eta=e$ and $\eta(0)=0$.

For given $f\in\mathcal F_{L,b}$ and $y^*\in\mathbb R^m$, the map
$u\mapsto f(y^*,u)$ is strongly monotone and coercive, and therefore admits
a unique zero $u^*\in\mathbb R^m$. For any gain pair with $k_i\ne0$, the
 augmented closed-loop equilibrium is
$(x,\eta)=(y^*,u^*/k_i)$.

We distinguish the trajectory-defined system-level PI region from its
common-quadratic inner region. The following definitions introduce
these two regions.

\begin{definition}[System-level robust PI gain region]
\label{def:pi-uniform-region}
A gain pair $(k_p,k_i)\in\mathbb R^2$ with $k_i\ne0$ belongs to
$\Omega_{\rm pi}$ if there exist constants $M\ge1$ and $\lambda>0$,
depending only on the gains and the uncertainty bounds $(L,b)$, such that,
for every $f\in\mathcal F_{L,b}$, every constant reference, and every
initial condition, the closed-loop solution is forward
complete and satisfies
\begin{equation}
\label{eq:pi-uniform-estimate}
\norm{z(t)}
\le
Me^{-\lambda t}\norm{z(0)},
\qquad t\ge0,
\end{equation}
where
\[
z(t):=
\operatorname{col}\!\left(
\eta(t)-u^*/k_i,\,e(t)
\right),
\qquad
\eta(t):=\textstyle\int_0^t e(s)\,ds.
\]
The set $\Omega_{\rm pi}$ is called the
\emph{system-level robust PI gain region}.
\end{definition}

To compare common-quadratic
certification with system-level robustness, define
\[
\mathcal A_{\rm PI}(D,\Theta):=
\begin{bmatrix}
0&I_m\\
-k_i\Theta&D-k_p\Theta
\end{bmatrix},
\qquad
D,\Theta\in\mathbb R^{m\times m}.
\]

\begin{definition}[Common-quadratic PI gain region]
\label{def:pi-common-quadratic}
A gain pair $(k_p,k_i)\in\mathbb R^2$ is called
\emph{common-quadratically admissible} if there exist a matrix
$P=P^\top\succ0$ in $\mathbb R^{2m\times2m}$ and a constant $\alpha>0$
such that
\begin{equation*}
P\mathcal A_{\rm PI}(D,\Theta)
+\mathcal A_{\rm PI}(D,\Theta)^\top P
\preceq-\alpha I_{2m}
\end{equation*}
for any $D,\Theta\in\mathbb R^{m\times m}$ satisfying
\(
\norm D\le L, \
\Sym(\Theta)\succeq bI_m.
\)
Any such $P$ is called a \emph{common quadratic certificate}. The set of
all common-quadratically admissible PI gains is denoted by
$\Omega_{\rm pi}^{\rm cq}$ and called the
\emph{common-quadratic PI gain region}.
\end{definition}

The same state-dependent Jacobian argument used for
Proposition~\ref{prop:cq-certificate} shows that every common quadratic
certificate yields uniform exponential regulation over
$\mathcal F_{L,b}$. Hence,
 one can obtain the relationship
\(
\Omega_{\rm pi}^{\rm cq}
\subseteq
\Omega_{\rm pi}.
\)

\subsection{Boundary-linear equivalence and the exact PI gain region}

To characterize $\Omega_{\rm pi}$ exactly, we construct a
boundary-linear test family whose state and symmetric input Jacobians
attain the uncertainty bounds while its skew-symmetric input component
becomes arbitrarily large.

For $m\ge2$, define
\[
J_2:=
\begin{bmatrix}
0&-1\\
1&0
\end{bmatrix},\quad
J_m:=
\begin{bmatrix}
J_2 & 0_{2\times(m-2)}\\
0_{(m-2)\times2} & 0_{m-2}
\end{bmatrix}.
\]
For $m=2$, the zero-dimensional blocks in this display are absent, so
$J_m=J_2$.
Let $\mathbb N_0:=\{0,1,2,\ldots\}$. We introduce the countable
boundary-linear family
\begin{equation}
\label{eq:pi-boundary-family}
\begin{aligned}
\mathcal F_{\partial}^{m}
&:=\{f_\nu:\nu\in\mathbb N_0\},\\
f_\nu(x,u)&:=Lx+(bI_m+\nu J_m)u.
\end{aligned}
\end{equation}
Every member of $\mathcal F_{\partial}^{m}$ belongs to
$\mathcal F_{L,b}$ because
\[
\norm{\partial f_\nu/\partial x}=L,
\qquad
\Sym(\partial f_\nu/\partial u)=bI_m.
\]
Thus \eqref{eq:pi-boundary-family} is a subset of the nonlinear
uncertainty class $\mathcal F_{L,b}$ while retaining an unbounded skew-symmetric direction.

The inclusion
$\mathcal F_{\partial}^{m}\subseteq\mathcal F_{L,b}$ gives only one
direction; the nontrivial question is whether regulation of this
boundary family also recovers the full nonlinear gain region.

To compare regulation over the boundary-linear family with that over the
full nonlinear class, we introduce the following terminology. Consider
system~\eqref{sys:first-order} under the PI controller with
$(k_p,k_i)\in\mathbb R^2$ and $k_i\ne0$. For any subclass
$\mathcal G\subseteq\mathcal F_{L,b}$, the PI controller is said to
\emph{uniformly exponentially regulate} $\mathcal G$ if there exist
constants $M\ge1$ and $\lambda>0$, independent of the particular
$f\in\mathcal G$, reference, and initial condition, such that
\eqref{eq:pi-uniform-estimate} holds.

\begin{theorem}[Boundary-linear equivalence and exact PI region]
\label{thm:pi-exact-uniform}
Suppose $m\ge2$. Then the following statements are equivalent:
\begin{itemize}
\item[\rm i)] The PI controller uniformly exponentially regulates the full
nonlinear class $\mathcal F_{L,b}$.

\item[\rm ii)] The PI controller uniformly exponentially regulates the
boundary-linear family $\mathcal F_{\partial}^{m}$.

\item[\rm iii)] The gains satisfy
\begin{equation}
\label{eq:pi-exact-region}
k_p>0,\qquad
k_i>0,\qquad
k_i+Lk_p<bk_p^2.
\end{equation}
\end{itemize}
\end{theorem}

\begin{remark}\normalfont
Theorem~\ref{thm:pi-exact-uniform} reveals two distinct forms of
losslessness. First, the equivalence of statements~(i) and~(ii) shows that
the countable boundary-linear family $\mathcal F_{\partial}^{m}$ is a
\emph{lossless test family}: uniform exponential regulation over this family
recovers exactly the gain region for the full nonlinear uncertainty class.
This boundary-reduction principle is reminiscent of Kharitonov's theorem
for interval polynomials, although the present
equivalence concerns uniform exponential regulation of a nonlinear MIMO
class through a countable boundary family.
Second, the proof of (iii) $\Rightarrow $ (i) constructs a common quadratic
certificate for every gain pair satisfying the
condition in~(iii). Since condition~(iii) exactly characterizes
$\Omega_{\rm pi}$, this gives
\(
\Omega_{\rm pi}\subseteq\Omega_{\rm pi}^{\rm cq}.
\)
Together with the general reverse inclusion, we obtain
\begin{equation}
\label{eq:pi-cq-system-equality}
\Omega_{\rm pi}^{\rm cq}
=
\Omega_{\rm pi}
=
\left\{
(k_p,k_i)\in\mathbb R_{>0}^2:
k_i+Lk_p<bk_p^2
\right\}.
\end{equation}
Thus, losslessness occurs at two different levels:
$\mathcal F_{\partial}^{m}$ is lossless as a test family for the nonlinear
uncertainty class, while the common-quadratic class is lossless as a
certificate class for the system-level PI gain region.
\end{remark}

\begin{remark}\normalfont
For $m\ge2$, the exact condition in
Theorem~\ref{thm:pi-exact-uniform} improves the sufficient condition of
Zhao and Guo~\cite{ZhaoGuo2022} by removing the additional term
$L^2/(4b)$.  As illustrated in
Figure~\ref{fig:pi-region-comparison}, the admissible upper bound on $k_i$
is increased by $L^2/(4b)$ for each fixed $k_p$.
For $m=1$, the same common-quadratic sufficiency construction remains valid,
whereas the admissible linear system $f(x,u)=Lx+bu$ yields the necessary
condition $bk_p>L$. Consequently,
\begin{align*}
&\bigl\{(k_p,k_i)\in\mathbb R_{>0}^2:
k_i+Lk_p<bk_p^2\bigr\}
\subseteq
\Omega_{\rm pi}^{\rm cq}
\subseteq
\Omega_{\rm pi}\\
&\hspace{25mm}\subseteq
\bigl\{(k_p,k_i)\in\mathbb R_{>0}^2:
bk_p>L\bigr\}.
\end{align*}
The outer set is the exact pointwise robust PI region obtained
in~\cite{ZhaoGuo2022}. Since pointwise asymptotic regulation is weaker than
uniform exponential regulation, these inclusions bracket, but do not
determine, the scalar system-level PI gain region.
\end{remark}

\begin{figure}[t]
\centering
\includegraphics[width=0.72\linewidth]
{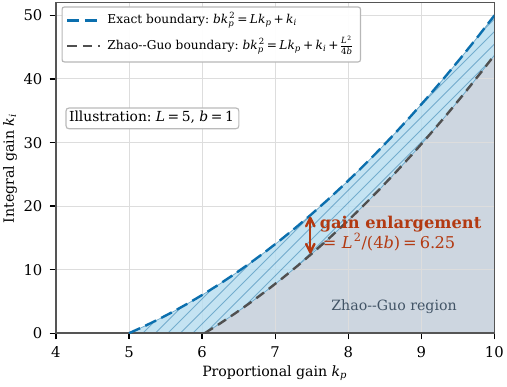}
\caption{Comparison of the PI gain regions for $L=5$ and $b=1$.
The gray area is derived in
\cite{ZhaoGuo2022}, whereas the hatched area represents the
additional gains admitted by the exact condition. The vertical
difference between the two boundaries is
$L^2/(4b)=6.25$.}
\label{fig:pi-region-comparison}
\end{figure}

\section{Scope and Limits of Common-Quadratic Losslessness}
\label{sec:cq-losslessness-counterexample}
The preceding PID and PI results show that, for the two nonlinear
uncertainty classes considered above and for $m\ge2$, common-quadratic
certification is lossless:
\[
\Omega_{\rm pid}^{\rm cq}=\Omega_{\rm pid},
\qquad
\Omega_{\rm pi}^{\rm cq}=\Omega_{\rm pi}.
\]
This losslessness is, however, specific to those uncertainty classes.
As the following example shows, even a one-parameter linear family over
a compact interval can have a system-level PID gain region that is
strictly larger than its common-quadratic counterpart.

Fix $L>0$ and consider the second-order linear family
\begin{equation}
\label{eq:cq-counterexample-plant}
\dot x_1=x_2,\qquad
\dot x_2=a x_2+u,
\end{equation}
where $x_1,x_2,u\in\mathbb R$ and $a\in[0,L]$ is an unknown but time-invariant parameter. Apply the PID controller
\eqref{eq:pid}. Since the equilibrium input associated with a constant
reference is zero, the augmented error state
$z=\operatorname{col}(\eta,e,\dot e)$ satisfies
\begin{equation}
\label{eq:cq-counterexample-matrix}
\dot z=A(a)z,
\qquad
A(a):=
\begin{bmatrix}
0&1&0\\
0&0&1\\
-k_i&-k_p&a-k_d
\end{bmatrix}.
\end{equation}
For this family, let $\mathcal R_{\rm sys}(L)$ denote the set of PID
gains for which the uniform exponential requirement in
Definition~\ref{def:pid-system-level} holds uniformly over all fixed
$a\in[0,L]$.
 Let
\[
\mathcal R_{\rm cq}(L)
:=
\left\{
(k_p,k_i,k_d):
\exists\,P\succ0\ \text{such that for all } a\in[0,L],\
\ PA(a)+A(a)^\top P\prec0
\right\}
\]
denote its common-quadratic PID gain region.

\begin{proposition}\normalfont
\label{prop:cq-losslessness-counterexample}
For every $L>0$,
\(
\mathcal R_{\rm cq}(L)
\subsetneq
\mathcal R_{\rm sys}(L).
\)
\end{proposition}

\begin{proof}
The inclusion
$\mathcal R_{\rm cq}(L)\subseteq\mathcal R_{\rm sys}(L)$
follows directly from the Lyapunov inequality. Indeed, for a fixed common
$P\succ0$, the compactness of $[0,L]$ and continuity of $A(a)$ give a
uniform positive margin in
$-\bigl(PA(a)+A(a)^\top P\bigr)$, and hence uniform exponential decay.

It remains to prove strictness. Consider the gain triple
\begin{equation}
\label{eq:cq-counterexample-gains}
(k_p,k_i,k_d)
=
\left(
\frac{3}{2}L^2,
\frac{5}{2}L^3,
3L
\right).
\end{equation}
For any $a\in[0,L]$, the characteristic polynomial of $A(a)$ is
\[
p_a(s)
=
s^3+(3L-a)s^2+\frac{3}{2}L^2s+\frac{5}{2}L^3.
\]
All its coefficients are positive and
\[
(3L-a)\frac{3}{2}L^2
\ge
3L^3
>
\frac{5}{2}L^3.
\]
The Routh--Hurwitz criterion therefore shows that $A(a)$ is Hurwitz for
every $a\in[0,L]$. Since the matrix family is continuous and the uncertainty
interval is compact, the corresponding semigroups admit uniform constants
$M\ge1$ and $\lambda>0$. Thus the gain triple
\eqref{eq:cq-counterexample-gains} belongs to $\mathcal R_{\rm sys}(L)$.

We now show that the same gain triple does not belong to
$\mathcal R_{\rm cq}(L)$. Let $A_0:=A(0)$ and $A_L:=A(L)$. A direct
calculation gives
\begin{equation}
\label{eq:cq-counterexample-product}
\det(\mu I_3-A_LA_0)
=
\frac14(\mu+L^2)
\bigl(4\mu^2-16L^2\mu-25L^4\bigr).
\end{equation}
Hence $A_LA_0$ has the negative real eigenvalue $-L^2$. Let
$v\ne0$ be a corresponding real eigenvector and set $w:=A_0v$. If a common
$P=P^\top\succ0$ existed, then the Lyapunov inequality for $A_0$ would give
\(
2v^\top PA_0v=2v^\top Pw<0,
\)
so $v^\top Pw<0$. On the other hand, since
$A_Lw=A_LA_0v=-L^2v$, the Lyapunov inequality for $A_L$ would require
\[
0>
2w^\top PA_Lw
=
-2L^2w^\top Pv
>0,
\]
a contradiction. Therefore no common quadratic Lyapunov matrix exists for
the endpoint pair, and consequently none exists for the whole interval.
Thus
\[
\left(
\frac{3}{2}L^2,
\frac{5}{2}L^3,
3L
\right)
\in
\mathcal R_{\rm sys}(L)
\setminus
\mathcal R_{\rm cq}(L),
\]
which proves the proposition.
\end{proof}

Thus, the coincidence between the common-quadratic and system-level
gain regions established above is a structural property of the two
nonlinear uncertainty classes, rather than an automatic consequence of
robust stability over a compact parametric family.

\section{Proofs of the main results}
\label{sec:proof-main}

The proofs follow the structure of the main results. We first establish
the exact common-quadratic PID characterization, then derive the exact
linear-subclass outer region. Their coincidence closes the PID
inner--outer gap. Finally, we prove the boundary-linear equivalence and
the exact PI characterization.

\subsection{Common-quadratic PID characterization}

The proof of Theorem \ref{thm:main} begins with two structural reductions.  The first shows that the
isotropy of the uncertainty class makes the Kronecker form of the
certificate lossless.  The second derives positivity of the gains and a
normalization that fixes the last column of the reduced certificate.

\begin{lemma}
\label{lem:symmetry-reduction}
A gain triple is common-quadratically admissible if and only if
\eqref{eq:robust-matrix-inequality} admits a certificate of the form
$P_0\otimes I_m$, where $P_0=P_0^\top\succ0$ belongs to
$\mathbb R^{3\times3}$.
Here $\otimes$ denotes the Kronecker product: if $A=(a_{ij})$, then
$A\otimes B$ is the block matrix whose $(i,j)$ block is $a_{ij}B$.
In particular, the $(i,j)$ block of $P_0\otimes I_m$ is
$(P_0)_{ij}I_m$.
\end{lemma}

\begin{proof}
Only necessity requires proof.  Let $P=P^\top\succ0$ be an arbitrary
certificate with margin $\alpha>0$, and let $\mathcal G_m^{\pm}$ denote
the signed-permutation group
\begin{align*}
\mathcal G_m^{\pm}:=
\bigl\{E_{\varepsilon}\mathsf P:\;&E_{\varepsilon}=
\operatorname{diag}(\varepsilon_1,\ldots,\varepsilon_m),
\varepsilon_j\in\{-1,1\},\\[-1mm]
&\ \mathsf P\text{ is an }m\times m\text{ permutation matrix}\bigr\}.
\end{align*}
Equivalently, each $U\in\mathcal G_m^{\pm}$ has exactly one nonzero entry
in every row and column, and that entry is either $1$ or $-1$.  Thus
$\mathcal G_m^{\pm}$ is a finite subgroup of the orthogonal group, with
$U^\top U=I_m$ and $|\mathcal G_m^{\pm}|=2^m m!$.  For
$U\in\mathcal G_m^{\pm}$, set $T_U:=I_3\otimes U$.

The uncertainty set is invariant under the simultaneous conjugation
\[
(D_1,D_2,\Theta)
\longmapsto
(UD_1U^\top,UD_2U^\top,U\Theta U^\top),
\]
and the corresponding closed-loop matrices satisfy
\[
\mathcal A(UD_1U^\top,UD_2U^\top,U\Theta U^\top)
=T_U\mathcal A(D_1,D_2,\Theta)T_U^\top.
\]
Applying the inequality for $P$ to the conjugated uncertainty triple and
then conjugating the result by $T_U$ shows that $T_U^\top PT_U$ is also a
certificate with margin $\alpha$.  Hence the group average
\[
\overline P:=
\frac{1}{|\mathcal G_m^{\pm}|}
\sum_{U\in\mathcal G_m^{\pm}}T_U^\top PT_U
\]
is positive definite and remains a certificate with the same margin.

The group property also gives
$T_U^\top\overline PT_U=\overline P$ for every
$U\in\mathcal G_m^{\pm}$.  Write $\overline P$ in $m\times m$ blocks.
Invariance under coordinate sign changes forces each block to be diagonal,
whereas invariance under coordinate permutations forces all of its diagonal
entries to coincide.  Every block is therefore a scalar multiple of $I_m$,
so
\(
\overline P=P_0\otimes I_m
\)
for some $P_0=P_0^\top\succ0$.  The converse is immediate.
\end{proof}

Lemma~\ref{lem:symmetry-reduction} shows that restricting the certificate
to $P_0\otimes I_m$ entails no loss of admissible gains.  The next lemma
shows that feasibility also forces positivity of the gains and permits a
useful normalization of $P_0$.

\begin{lemma}
\label{lem:automatic-normalization}
If $(k_p,k_i,k_d)\in\Omega_{\rm pid}^{\rm cq}$, then
\[
k_i>0,\qquad k_p>0,\qquad k_d>0,
\qquad bk_pk_d>k_i.
\]
Besides, there is a certificate $P_0\otimes I_m$
such that
\[
P_0(0,0,1)^\top=(k_i,k_p,k_d)^\top.
\]
\end{lemma}

\begin{proof}

First, as shown in the proof of Proposition~\ref{prop:cq-certificate}, the nominal
choice $D_1=D_2=0$ and $\Theta=bI_m$, together with the Routh--Hurwitz
criterion, gives $k_i,k_p,k_d>0$ and $bk_pk_d>k_i$.

By Lemma~\ref{lem:symmetry-reduction}, there exists a certificate
$\overline P_0\otimes I_m$, where
$\overline P_0=\overline P_0^\top\succ0$, with some margin $\alpha>0$.
It remains to normalize this certificate. Denote $e_3=(0,0,1)^\top,
\kappa=(k_i,k_p,k_d)^\top,$ and put
$p:=\overline P_0e_3$.  For $\tau\ge0$, choose
$D_1=D_2=0$ and $\Theta=(b+\tau)I_m$.  The robust inequality becomes
\[
\bigl[H_b-\tau(p\kappa^\top+\kappa p^\top)\bigr]\otimes I_m
\preceq-\alpha I_{3m},
\qquad \tau\ge0,
\]
where
\[
A_b:=\begin{bmatrix}0&1&0\\0&0&1\\-bk_i&-bk_p&-bk_d\end{bmatrix},
\qquad H_b:=\overline P_0A_b+A_b^\top\overline P_0
\]
are independent of
$\tau$.  Consequently,
\[
p\kappa^\top+\kappa p^\top\succeq0;
\]
otherwise, evaluation along a negative direction of this matrix would
contradict the preceding inequality for all sufficiently large $\tau$.

Since $\kappa\ne0$, the last positive-semidefiniteness condition implies
$p=c\kappa$ for some $c\ge0$.  To see this, write
$p=c\kappa+r$, where $r^\top\kappa=0$.  If $r\ne0$, the restriction of
$p\kappa^\top+\kappa p^\top$ to
$\operatorname{span}\{\kappa,r\}$ has negative determinant, which is
impossible for a positive-semidefinite matrix.  Thus $r=0$, and positive
semidefiniteness then gives $c\ge0$.  Furthermore,
\[
e_3^\top p=e_3^\top\overline P_0e_3>0,
\qquad
e_3^\top\kappa=k_d>0,
\]
so $c>0$.  Defining $P_0:=c^{-1}\overline P_0$ gives another common
quadratic certificate, now with margin $\alpha/c$, and satisfies
$P_0e_3=\kappa$.
\end{proof}

\begin{remark}\normalfont
The preceding lemmas reduce the certificate-feasibility problem substantially.
Lemma~\ref{lem:symmetry-reduction} restricts the search losslessly from
$P\in\mathbb S_{++}^{3m}$ to $P_0\in\mathbb S_{++}^{3}$, while
Lemma~\ref{lem:automatic-normalization} gives
\[
P_0=
\begin{bmatrix}
p_{11}&p_{12}&k_i\\
p_{12}&p_{22}&k_p\\
k_i&k_p&k_d
\end{bmatrix}.
\]
Hence, for fixed PID gains, only $p_{11}$, $p_{12}$, and $p_{22}$ remain
to be eliminated. For $m\ge2$, the rank-two Gram reduction and minimax
argument perform this elimination exactly and yield the auxiliary scalars
$\mu$ and $\rho$ in \eqref{eq:exact-cq-region}. These scalars arise from
the reduction itself rather than from an a priori parametrization of $P_0$.
The only step requiring $m\ge2$ is the rank-two Gram realizability.
\end{remark}

\begin{proof}[Proof of Theorem~\ref{thm:main}]
Set $\kappa:=(k_i,k_p,k_d)^\top$ and
\[
F:=
\begin{bmatrix}
0&1&0\\
0&0&1\\
0&0&0
\end{bmatrix}.
\]

\emph{Part I: exact characterization for $m=1$.}
With $m=1$ understood, we prove that membership in
$\Omega_{\rm pid}^{\rm cq}$ is equivalent to
\eqref{eq:four-vertex-condition}.

First, let $(k_p,k_i,k_d)\in\Omega_{\rm pid}^{\rm cq}$. By
Lemma~\ref{lem:automatic-normalization}, the gains are positive and a
certificate may be normalized so that $P_0e_3=\kappa$. For admissible
scalars $d_1,d_2,\theta$, define
\[
\mathcal A^{\rm sc}(d_1,d_2,\theta):=
\begin{bmatrix}
0&1&0\\
0&0&1\\
-\theta k_i&d_1-\theta k_p&d_2-\theta k_d
\end{bmatrix}
\]
and
\[
\begin{aligned}
Q(d_1,d_2,\theta)
:=-P_0\mathcal A^{\rm sc}(d_1,d_2,\theta)
-\mathcal A^{\rm sc}(d_1,d_2,\theta)^\top P_0.
\end{aligned}
\]
Since
\begin{equation*}
\mathcal A^{\rm sc}(d_1,d_2,\theta)
=
\mathcal A^{\rm sc}(d_1,d_2,b)
-(\theta-b)e_3\kappa^\top,
\end{equation*}
the normalization $P_0e_3=\kappa$ gives
\begin{equation}
\label{eq:scalar-theta-reduction}
Q(d_1,d_2,\theta)
=
Q(d_1,d_2,b)
+2(\theta-b)\kappa\kappa^\top.
\end{equation}
Thus \(\theta=b\) is the least favorable value.

At $\theta=b$, the matrix $Q(d_1,d_2,b)$ is affine in $(d_1,d_2)$.
Every point of $[-L_1,L_1]\times[-L_2,L_2]$ is a convex combination of
its four vertices. Hence positivity for every
admissible $(d_1,d_2)$ is equivalent to
\[
-\Bigl(
P_0\mathcal A_{st}^{\rm sc}
+(\mathcal A_{st}^{\rm sc})^\top P_0
\Bigr)\succ0,
\qquad (s,t)\in\mathcal V^2.
\]
Therefore every gain in $\Omega_{\rm pid}^{\rm cq}$ satisfies
\eqref{eq:four-vertex-condition}.

Conversely, suppose a gain triple in $\mathbb R^3$ satisfies
\eqref{eq:four-vertex-condition}, and let $P_0$ be the corresponding
four-vertex certificate. Define the vertex average
\begin{equation*}
A_b:=
\begin{bmatrix}
0&1&0\\
0&0&1\\
-bk_i&-bk_p&-bk_d
\end{bmatrix}
=\frac14\sum_{(s,t)\in\mathcal V^2}\mathcal A_{st}^{\rm sc}.
\end{equation*}
Averaging the four strict LMIs therefore yields
$P_0A_b+A_b^\top P_0\prec0$, so $A_b$ is Hurwitz. Its characteristic polynomial is
$s^3+bk_ds^2+bk_ps+bk_i$, and the Routh--Hurwitz criterion yields
\[
k_i>0,\qquad k_p>0,\qquad k_d>0,
\qquad bk_pk_d>k_i.
\]
Thus positivity is a consequence of the four-vertex condition rather than
an additional assumption. Now set
\[
\alpha_1:=
\min_{(s,t)\in\mathcal V^2}
\lambda_{\min}\!\left(
-P_0\mathcal A_{st}^{\rm sc}
-(\mathcal A_{st}^{\rm sc})^\top P_0
\right)>0.
\]
Convexity and \eqref{eq:scalar-theta-reduction} yield
$Q(d_1,d_2,\theta)\succeq\alpha_1I_3$ for all $|d_j|\le L_j$ and
$\theta\ge b$. Thus $P_0$ satisfies the uniform robust matrix inequality,
so the gain belongs to $\Omega_{\rm pid}^{\rm cq}$. This proves the scalar
equivalence.

\emph{Part II: sufficiency of the explicit multivariable conditions.}
Let $(k_p,k_i,k_d)$ satisfy the conditions on the right-hand side of
\eqref{eq:exact-cq-region}. Choose the corresponding
$\mu,\rho\in(0,1)$ and define
\[
\sigma:=2b(1-\mu),\qquad
\varepsilon_1:=2b\mu\rho,\qquad
\varepsilon_2:=2b\mu(1-\rho).
\]
Then $\sigma+\varepsilon_1+\varepsilon_2=2b$.
Put
\begin{align*}
A_0
&:=\sigma(k_p^2-2k_ik_d)-\frac{L_1^2}{\varepsilon_1}
 =\sigma(\Pi_{\mu,\rho}-2k_ik_d),\\
C_0
&:=\sigma k_d^2-2k_p-\frac{L_2^2}{\varepsilon_2}
 =\sigma\Delta_{\mu,\rho}.
\end{align*}
The defining inequalities in \eqref{eq:exact-cq-region} give
$C_0>0$ and $A_0+2k_i\sqrt{\sigma C_0}>0$. Hence one may choose
$\xi\in\mathbb R$ such that
\begin{equation}
\label{eq:xi-choice}
\max\left\{
-\frac{A_0}{2},
-k_i\sqrt{\sigma C_0}
\right\}
<\xi<
k_i\sqrt{\sigma C_0}.
\end{equation}

Define
\[
P_0:=
\begin{bmatrix}
\sigma k_ik_p&\sigma k_ik_d-\xi&k_i\\
\sigma k_ik_d-\xi&\sigma k_pk_d-k_i&k_p\\
k_i&k_p&k_d
\end{bmatrix}.
\]
Set $P:=P_0\otimes I_m$. Then $P_0e_3=\kappa$. For every admissible triple and every
$z=\operatorname{col}(z_0,z_1,z_2)$, put
$q:=(\kappa^\top\otimes I_m)z$.
Set
\[
\mathcal Q_b:=-(P_0F+F^\top P_0)+2b\kappa\kappa^\top.
\]
The identity $P_0e_3=\kappa$ gives
\begin{align*}
&-z^\top\!\left[
P\mathcal A(D_1,D_2,\Theta)
+\mathcal A(D_1,D_2,\Theta)^\top P
\right]z\\
&=z^\top(\mathcal Q_b\otimes I_m)z
+2q^\top\bigl(\Sym(\Theta)-bI_m\bigr)q\\
&\quad-2q^\top D_1z_1-2q^\top D_2z_2.
\end{align*}
Applying Young's inequality to the last two terms and using
$\Sym(\Theta)\succeq bI_m$ yields
\[
-2q^\top D_jz_j
\ge-\varepsilon_j\norm q^2
-\frac{L_j^2}{\varepsilon_j}\norm{z_j}^2,
\qquad j=1,2,
\]
and hence
\begin{align}
&-z^\top\Bigl[
P\mathcal A(D_1,D_2,\Theta)
+\mathcal A(D_1,D_2,\Theta)^\top P
\Bigr]z\notag\\
&\quad\ge z^\top(\widehat Q\otimes I_m)z,
\label{eq:constructed-robust-bound}
\end{align}
where
\[
\widehat Q:=
\begin{bmatrix}
\sigma k_i^2&0&\xi\\
0&A_0+2\xi&0\\
\xi&0&C_0
\end{bmatrix}.
\]
By \eqref{eq:xi-choice}, $A_0+2\xi>0$ and
$\xi^2<\sigma k_i^2C_0$, and therefore $\widehat Q\succ0$.

It remains to verify $P_0\succ0$. Since
$k_d^2-\Delta_{\mu,\rho}\ge k_p/[b(1-\mu)]$, we have
\begin{align*}
k_d-\sqrt{\Delta_{\mu,\rho}}
=\frac{k_d^2-\Delta_{\mu,\rho}}
{k_d+\sqrt{\Delta_{\mu,\rho}}}
>\frac{k_p}{2bk_d}.
\end{align*}
Hence
$2bk_pk_d(k_d-\sqrt{\Delta_{\mu,\rho}})>k_p^2>
\Pi_{\mu,\rho}$. Together with the upper bound on $k_i$, this gives
$k_i<bk_pk_d$;
hence $s^3+bk_ds^2+bk_ps+bk_i$ is Hurwitz. Let
$A_b:=(F-be_3\kappa^\top)\otimes I_m$ and
$Q_b^{\rm nom}:=-(PA_b+A_b^\top P)$.
Taking $D_1=D_2=0$ and $\Theta=bI_m$ in
\eqref{eq:constructed-robust-bound} gives
\[
Q_b^{\rm nom}\succeq\widehat Q\otimes I_m\succ0.
\]
Since $A_b$ is Hurwitz, uniqueness of the Lyapunov equation gives
\[
P=\int_0^\infty
e^{A_b^\top t}Q_b^{\rm nom}e^{A_bt}\,dt\succ0.
\]
Hence $P_0\succ0$. Equation
\eqref{eq:constructed-robust-bound} now supplies a uniform negative margin,
so the gain is common-quadratically admissible. Thus every gain satisfying
the explicit conditions in \eqref{eq:exact-cq-region} belongs to
$\Omega_{\rm pid}^{\rm cq}$, in fact for every $m\ge1$.

\emph{Part III: necessity of the explicit multivariable conditions.}
Let $(k_p,k_i,k_d)\in\Omega_{\rm pid}^{\rm cq}$. By
Lemma~\ref{lem:automatic-normalization}, the gains are positive and we may
choose a normalized certificate $P_0$ satisfying $P_0e_3=\kappa$. Define
$\mathcal Q_b:=-(P_0F+F^\top P_0)+2b\kappa\kappa^\top$ and, for
$z=\operatorname{col}(z_0,z_1,z_2)$, set
$q:=(\kappa^\top\otimes I_m)z$, and abbreviate
$\mathcal A=\mathcal A(D_1,D_2,\Theta)$.  The normalization
$P_0e_3=\kappa$ gives the exact identity
\begin{align*}
&-z^\top\!\left[
  (P_0\otimes I_m)\mathcal A
  +\mathcal A^\top(P_0\otimes I_m)
  \right]z\\
&\quad=z^\top(\mathcal Q_b\otimes I_m)z\\
&\qquad+2q^\top\bigl(\Sym(\Theta)-bI_m\bigr)q\\
&\qquad-2q^\top D_1z_1-2q^\top D_2z_2.
\end{align*}
Thus the skew-symmetric part of $\Theta$ cancels and $\Theta=bI_m$ is
least favorable.  Minimizing independently over the two norm balls, and
using the uniform margin of the normalized certificate, gives some
$\gamma>0$ such that, for every $z\ne0$,
\begin{equation}
z^\top(\mathcal Q_b\otimes I_m)z
-2L_1\norm q\norm{z_1}
-2L_2\norm q\norm{z_2}
\ge\gamma\norm z^2.
\label{eq:exact-worst-form}
\end{equation}
Indeed, equality in each norm bound is attained, whenever the relevant
vectors are nonzero, by
\[
D_j
=
L_j\frac{q}{\norm q}
\frac{z_j^\top}{\norm{z_j}},
\qquad j=1,2.
\]

Let $Z=[z_0\ z_1\ z_2]\in\mathbb R^{m\times3}$ and let
$G:=Z^\top Z$ be its Gram matrix, so
$G_{ij}=z_{i-1}^\top z_{j-1}$ for $i,j=1,2,3$.
Write
\[
\mathcal K:=\{G\in\mathbb S_+^3:\operatorname{tr}G=1\},
\]
where $\operatorname{tr}$ is the sum of the diagonal entries. Below,
$\langle A,B\rangle:=\operatorname{tr}(A^\top B)$ is the Frobenius
inner product, $\operatorname{rank}$ is matrix rank, and
$e_2=(0,1,0)^\top$, $e_3=(0,0,1)^\top$ are the second and third
standard basis vectors of $\mathbb R^3$.
Define
\[
r(G):=\kappa^\top G\kappa,\qquad
s_1(G):=e_2^\top Ge_2,\qquad
s_2(G):=e_3^\top Ge_3,
\]
and
\[
\phi(G):=\langle\mathcal Q_b,G\rangle
-2L_1\sqrt{r(G)s_1(G)}
-2L_2\sqrt{r(G)s_2(G)}.
\]
Equation~\eqref{eq:exact-worst-form} gives $\phi(G)\ge\gamma$ for every
realizable $G\in\mathcal K$ with $\operatorname{rank}G\le m$.

Indeed, $\operatorname{tr}G=\norm z^2$, $r(G)=\norm q^2$, and
$s_j(G)=\norm{z_j}^2$. Both sides of \eqref{eq:exact-worst-form} are
homogeneous of degree two in $z$, so restricting to $\norm z=1$ is
lossless. Also, $\mathcal K$ is compact: its eigenvalues are nonnegative
and sum to one, and hence $\norm G\le1$.

For an arbitrary $G\in\mathcal K$, consider the compact convex fiber
\[
\mathcal K_G:=\{H\in\mathcal K:\ r(H)=r(G),\quad
s_1(H)=s_1(G),\quad s_2(H)=s_2(G)\}.
\]

The fiber is nonempty because it contains $G$. The linear functional
$H\mapsto\langle\mathcal Q_b,H\rangle$ attains its minimum on this
compact fiber. Its minimizers form a nonempty compact face, so we may
choose a minimizer $H_*$ that is an extreme point of $\mathcal K_G$.
To justify the last assertion, an extreme point of the minimizing face
is also extreme in the fiber: if it is a nontrivial convex combination
of two fiber points, linearity forces both points to be minimizers.

We claim that $\operatorname{rank}H_*\le2$. Otherwise $H_*\succ0$.
Since $\dim\mathbb S^3=6$, the following four homogeneous linear
equations have a nonzero solution $X\in\mathbb S^3$:
\[
\operatorname{tr}X=0,\qquad
\kappa^\top X\kappa=0,\qquad
e_2^\top Xe_2=0,\qquad e_3^\top Xe_3=0.
\]
For $0<\tau<\lambda_{\min}(H_*)/\norm X$, both $H_*+\tau X$ and
$H_*-\tau X$ are positive definite and satisfy all fiber constraints.
They are distinct, and their midpoint is $H_*$, contradicting
extremality. This proves the claim without adding the objective value
as an extra affine constraint.

All nonlinear terms in $\phi$ are constant on the fiber. Consequently,
$\phi(H_*)\le\phi(G)$. If $r_*:=\operatorname{rank}H_*\le2$, its
spectral decomposition gives $H_*=Y^\top Y$ with
$Y\in\mathbb R^{r_*\times3}$. For $m\ge2$, appending zero rows gives
a matrix in $\mathbb R^{m\times3}$ with the same Gram matrix. Thus
$H_*$ is realizable by three vectors in $\mathbb R^m$, and
$\phi(G)\ge\phi(H_*)\ge\gamma$ for every $G\in\mathcal K$.
Since $\phi$ is continuous on $\mathcal K$, it follows that

\begin{equation}
\label{eq:positive-gram-margin}
\delta:=\min_{G\in\mathcal K}\phi(G)>0.
\end{equation}

For $\varepsilon_1,\varepsilon_2>0$, define
\[
\mathcal Q_{\varepsilon}
:=
\mathcal Q_b
-(\varepsilon_1+\varepsilon_2)\kappa\kappa^\top
-\frac{L_1^2}{\varepsilon_1}e_2e_2^\top
-\frac{L_2^2}{\varepsilon_2}e_3e_3^\top.
\]
Using
\[
2L\sqrt{rs}
=
\inf_{\varepsilon>0}
\left(
\varepsilon r+\frac{L^2}{\varepsilon}s
\right),
\]
we obtain $\phi(G)=\sup_{\varepsilon_1,\varepsilon_2>0}
\langle\mathcal Q_{\varepsilon},G\rangle$.

The scalar identity holds for all $r,s\ge0$. For $r,s>0$, equality is
attained at $\varepsilon=L\sqrt{s/r}$; for $r=0<s$ the infimum is
approached as $\varepsilon\to\infty$, whereas for $s=0<r$ it is
approached as $\varepsilon\downarrow0$. If both vanish it is zero
for every $\varepsilon$. Thus the identity remains valid even when
$r(G)$ or $s_j(G)$ vanishes; no pointwise attainment is assumed.

Together with \eqref{eq:positive-gram-margin}, this identifies $\delta$ as
a minimax value.

More explicitly, put $\mathcal E:=(0,\infty)^2$ and
$\Phi(G,\varepsilon):=\langle\mathcal Q_\varepsilon,G\rangle$.
For each $\varepsilon\in\mathcal E$, $\Phi(\cdot,\varepsilon)$ is
continuous and affine, hence convex, on the compact convex set
$\mathcal K$. For each $G\in\mathcal K$, it is continuous and concave
in $\varepsilon$, since
\[
\frac{\partial^2\Phi}{\partial\varepsilon_j^2}
=-\frac{2L_j^2s_j(G)}{\varepsilon_j^3}\le0,
\qquad
\frac{\partial^2\Phi}{\partial\varepsilon_1\partial\varepsilon_2}=0.
\]
The set $\mathcal E$ is convex. These are the hypotheses of Sion's
minimax theorem~\cite{Sion1958}; only $\mathcal K$ needs to be compact,
so the openness and unboundedness of $\mathcal E$ cause no difficulty.
Therefore

\begin{align*}
0<\delta
&=\min_{G\in\mathcal K}
\sup_{\varepsilon_1,\varepsilon_2>0}
\langle\mathcal Q_{\varepsilon},G\rangle\\
&=
\sup_{\varepsilon_1,\varepsilon_2>0}
\min_{G\in\mathcal K}
\langle\mathcal Q_{\varepsilon},G\rangle\\
&=
\sup_{\varepsilon_1,\varepsilon_2>0}
\lambda_{\min}(\mathcal Q_{\varepsilon}).
\end{align*}

For the last equality, diagonalize $\mathcal Q_\varepsilon$: its
inner product with any $G\in\mathcal K$ is a convex combination of
its eigenvalues, and $G=vv^\top$ for a unit minimum-eigenvalue
eigenvector attains the smallest one. Since $\delta>0$, the definition
of the supremum now supplies a \emph{finite} pair
$\varepsilon\in\mathcal E$ for which
$\lambda_{\min}(\mathcal Q_\varepsilon)>\delta/2>0$.
Attainment of the supremum is not needed. Crucially, this single pair
works for every $G$, rather than depending on the state vectors.

Put
\[
\sigma:=2b-\varepsilon_1-\varepsilon_2,\qquad
\mu:=\frac{\varepsilon_1+\varepsilon_2}{2b},\qquad
\rho:=\frac{\varepsilon_1}{\varepsilon_1+\varepsilon_2}.
\]
The $(1,1)$ entry of $\mathcal Q_{\varepsilon}$ is
$\sigma k_i^2>0$. Hence $\sigma>0$ and $\mu,\rho\in(0,1)$.
Write
\[
P_0=
\begin{bmatrix}
p_{11}&p_{12}&k_i\\
p_{12}&p_{22}&k_p\\
k_i&k_p&k_d
\end{bmatrix},
\qquad
\xi:=\sigma k_ik_d-p_{12},
\]
and define
\[
A_0
:=
\sigma(k_p^2-2k_ik_d)
-\frac{L_1^2}{\varepsilon_1},
\qquad
C_0
:=
\sigma k_d^2-2k_p
-\frac{L_2^2}{\varepsilon_2}.
\]
Direct calculation gives the following entry and principal submatrix,
where the paired subscript $\{1,3\},\{1,3\}$ selects rows and columns
1 and 3:
\[
(\mathcal Q_{\varepsilon})_{22}=A_0+2\xi,
\qquad
(\mathcal Q_{\varepsilon})_{\{1,3\},\{1,3\}}
=
\begin{bmatrix}
\sigma k_i^2&\xi\\
\xi&C_0
\end{bmatrix}.
\]
Therefore $C_0>0$, $A_0+2\xi>0$, and
$\xi^2<\sigma k_i^2C_0$. Since
$C_0=\sigma\Delta_{\mu,\rho}$ and
$A_0=\sigma(\Pi_{\mu,\rho}-2k_ik_d)$, we obtain
$\Delta_{\mu,\rho}>0$. Moreover, $k_p,k_d>0$ implies
$0<\Delta_{\mu,\rho}<k_d^2$. Finally, $A_0+2\xi>0$ and
$\xi<k_i\sqrt{\sigma C_0}$ give
\begin{align*}
0<A_0+2k_i\sqrt{\sigma C_0}
=\sigma\left[
\Pi_{\mu,\rho}
-2k_i\bigl(k_d-\sqrt{\Delta_{\mu,\rho}}\bigr)
\right].
\end{align*}
Since $k_i>0$ and $k_d-\sqrt{\Delta_{\mu,\rho}}>0$, it follows that
\[
\Pi_{\mu,\rho}>0,
\qquad
k_i<
\frac{\Pi_{\mu,\rho}}
{2\bigl(k_d-\sqrt{\Delta_{\mu,\rho}}\bigr)}.
\]
Thus $(k_p,k_i,k_d)$ satisfies the conditions on the right-hand side of
\eqref{eq:exact-cq-region}.

For completeness, the three necessary inequalities above are equivalent to the nonempty
interval
\[
\max\{-A_0/2,-k_i\sqrt{\sigma C_0}\}
<\xi<k_i\sqrt{\sigma C_0},
\qquad C_0>0.
\]
Conversely, choose $\xi$ in this interval and set
$p_{12}=\sigma k_ik_d-\xi$, $p_{11}=\sigma k_ik_p$, and
$p_{22}=\sigma k_pk_d-k_i$. These choices cancel the $(1,2)$ and
$(2,3)$ entries of $\mathcal Q_\varepsilon$, leaving exactly the
positive-definite matrix $\widehat Q$ constructed in Part II.
The separate Lyapunov-equation argument there establishes $P_0\succ0$.
Thus eliminating these entries and then $\xi$ introduces no additional
restriction or unverified positivity assumption on the certificate.

Combining Parts II and III proves the exact
multivariable representation. Finally, Part II with $m=1$, together with
Part I, shows that every gain in the multivariable region also satisfies
the scalar four-vertex condition.
\end{proof}

\subsection{Linear-subclass PID characterization}

\begin{proof}[Proof of Proposition~\ref{prop:pid-linear-outer}]
We first consider $m=1$. Write the admissible matrices as scalars
$d_1,d_2,\theta$, where
\[
|d_1|\le L_1,\qquad
|d_2|\le L_2,\qquad
\theta\ge b.
\]
The characteristic polynomial of the augmented PID closed loop is
\begin{equation*}
p_{d_1,d_2,\theta}(s)
=
s^3+(\theta k_d-d_2)s^2
+(\theta k_p-d_1)s+\theta k_i.
\end{equation*}
By the Routh--Hurwitz criterion, this polynomial is Hurwitz if and only if
\begin{equation*}
\begin{aligned}
&\theta k_i>0,\quad
\theta k_p-d_1>0,\quad
\theta k_d-d_2>0,\\
&(\theta k_p-d_1)(\theta k_d-d_2)>\theta k_i.
\end{aligned}
\end{equation*}

Necessity follows by choosing
$d_1=L_1$, $d_2=L_2$, and $\theta=b$. Indeed, the preceding conditions give
\[
\bar k_0>0,\qquad
\bar k_1>0,\qquad
\bar k_2>0,\qquad
\bar k_1\bar k_2>\bar k_0.
\]
Since $\bar k_1>0$ and
$\bar k_1\bar k_2>\bar k_0>0$ already imply $\bar k_2>0$, these conditions
are equivalent to those in \eqref{eq:pid-scalar-linear-region}.

Conversely, suppose that the conditions in
\eqref{eq:pid-scalar-linear-region} hold. They imply
$k_i,k_p,k_d>0$ and $\bar k_2>0$. Hence, for every admissible triple,
\[
\theta k_p-d_1\ge \bar k_1>0,
\qquad
\theta k_d-d_2\ge \bar k_2>0.
\]
Define
\begin{equation*}
g(\theta)
:=
\frac{(\theta k_p-L_1)(\theta k_d-L_2)}{\theta}.
\end{equation*}
Since $bk_p>L_1$ and $bk_d>L_2$,
\begin{equation*}
g'(\theta)
=
k_pk_d-\frac{L_1L_2}{\theta^2}
\ge
k_pk_d-\frac{L_1L_2}{b^2}
>0,
\qquad \theta\ge b.
\end{equation*}
Therefore,
\begin{align*}
\frac{(\theta k_p-d_1)(\theta k_d-d_2)}{\theta}
\ge g(\theta)
\ge g(b)
=
\frac{\bar k_1\bar k_2}{b}
>
\frac{\bar k_0}{b}
=k_i.
\end{align*}
All Routh--Hurwitz inequalities consequently hold for every admissible
$d_1,d_2,\theta$. This proves
\eqref{eq:pid-scalar-linear-region}.

We now turn to the case $m\ge2$.
We first note that the frequency inequality in
\eqref{eq:pid-linear-region} automatically implies the nominal
Routh--Hurwitz condition. Indeed, let
\[
\omega_0:=\sqrt{\frac{k_i}{k_d}}.
\]
Since the gains are positive,
\[
\chi(\omega_0)=k_p^2\omega_0^2,
\qquad
\sqrt{\chi(\omega_0)}=k_p\omega_0.
\]
Evaluating the frequency inequality at $\omega=\omega_0$ and dividing by
$k_p\omega_0^2>0$ gives
\[
bk_p-\omega_0^2>L_1+L_2\omega_0\ge0.
\]
Consequently,
\[
bk_p>\omega_0^2=\frac{k_i}{k_d},
\qquad\text{and hence}\qquad
bk_pk_d>k_i.
\]
Therefore the nominal polynomial
\[
s^3+bk_ds^2+bk_ps+bk_i
\]
is Hurwitz by the Routh--Hurwitz criterion. This provides the stable
starting point for the continuation argument below.

For the linear plant, the characteristic matrix of the augmented
PID closed loop is
\begin{equation}
\label{eq:pid-linear-characteristic}
\begin{aligned}
\mathcal P(s)
=s^3I_m+(k_d\Theta-D_2)s^2
+(k_p\Theta-D_1)s+k_i\Theta.
\end{aligned}
\end{equation}

We first prove necessity. Choosing
$D_1=D_2=0$ and $\Theta=bI_m$ gives the scalar characteristic polynomial
\begin{equation*}
p_0(s)=s^3+bk_ds^2+bk_ps+bk_i
\end{equation*}
in each channel. The Routh--Hurwitz criterion therefore requires
\[
k_i>0,\qquad k_p>0,\qquad k_d>0,
\qquad bk_pk_d>k_i.
\]

Suppose next that the frequency inequality in
\eqref{eq:pid-linear-region} fails at some $\omega>0$. In what follows,
$\mathrm i^2=-1$, $\overline\zeta$ is complex conjugation, and
$\Re\zeta$, $\Im\zeta$, and $|\zeta|$ denote the real part, imaginary
part, and modulus of $\zeta\in\mathbb C$. Put
\begin{equation*}
q_\omega:=k_i-k_d\omega^2+\mathrm{i}k_p\omega,
\qquad
|q_\omega|^2=\chi(\omega),
\end{equation*}
and define
\begin{equation*}
c_1:=\frac{\mathrm{i}\omega}{q_\omega},
\qquad
c_2:=-\frac{\omega^2}{q_\omega},
\qquad
\alpha_j:=L_j\frac{\overline{c_j}}{|c_j|}.
\end{equation*}
Finally, let
\begin{equation*}
\beta:=
\frac{\mathrm{i}\omega^3
+\mathrm{i}\omega\alpha_1-\omega^2\alpha_2}
{q_\omega}.
\end{equation*}
Then $|\alpha_j|=L_j$ and
\begin{equation*}
\Re\beta
=
\frac{k_p\omega^4}{\chi(\omega)}
+
\frac{L_1\omega+L_2\omega^2}
{\sqrt{\chi(\omega)}}
\ge b.
\end{equation*}
Moreover,
\begin{equation}
\label{eq:pid-complex-marginal}
q_\omega\beta+\omega^2\alpha_2
-\mathrm{i}\omega\alpha_1-\mathrm{i}\omega^3=0.
\end{equation}

For $\zeta\in\mathbb C$, define its realification by
\begin{equation*}
\Realify(\zeta):=
\begin{bmatrix}
\Re\zeta&-\Im\zeta\\
\Im\zeta&\Re\zeta
\end{bmatrix}.
\end{equation*}
When $m=2$, choose
\[
D_j=\Realify(\alpha_j),
\qquad
\Theta=\Realify(\beta).
\]
Then
\[
\norm{D_j}=|\alpha_j|\le L_j,
\qquad
\Sym(\Theta)=(\Re\beta)I_2\succeq bI_2.
\]
Identity \eqref{eq:pid-complex-marginal} implies that
$\mathcal P(\mathrm{i}\omega)$ is singular. Hence the corresponding
closed loop has an imaginary-axis characteristic root and is not Hurwitz.

For $m>2$, the same construction is embedded by taking
\begin{equation*}
D_j=
\operatorname{diag}\bigl(\Realify(\alpha_j),0_{m-2}\bigr),
\qquad
\Theta=
\operatorname{diag}\bigl(\Realify(\beta),bI_{m-2}\bigr).
\end{equation*}
Thus the frequency inequality is necessary for every $m\ge2$.

We now prove sufficiency.
 Suppose that
$(k_p,k_i,k_d)\in\mathbb R_{>0}^3$ satisfies the frequency condition on the
right-hand side of \eqref{eq:pid-linear-region}. Assume, to the contrary,
that an admissible linear plant has a characteristic root
$s=\mathrm{i}\omega$ for some $\omega>0$.

 Let
$v\in\mathbb C^m$ be a corresponding unit vector and set
\[
a_j:=v^*D_jv,\qquad
\beta:=v^*\Theta v.
\]
Then
\[
|a_j|\le L_j,\qquad \Re\beta\ge b.
\]
From $\mathcal P(\mathrm{i}\omega)v=0$ we obtain
\begin{equation*}
q_\omega\beta+\omega^2a_2
-\mathrm{i}\omega a_1-\mathrm{i}\omega^3=0.
\end{equation*}
Dividing by $q_\omega$ and taking real parts gives
\begin{equation*}
0
\ge
b-\frac{k_p\omega^4}{\chi(\omega)}
-\frac{L_1\omega+L_2\omega^2}
{\sqrt{\chi(\omega)}},
\end{equation*}
which contradicts \eqref{eq:pid-linear-region}. Hence no nonzero
imaginary-axis root is possible. A zero root is also impossible because
\[
\mathcal P(0)=k_i\Theta
\]
is nonsingular.

Finally, connect the nominal matrices $(0,0,bI_m)$ to any admissible
$(D_1,D_2,\Theta)$ through
\begin{equation*}
D_j(\tau)=\tau D_j,\qquad
\Theta(\tau)=(1-\tau)bI_m+\tau\Theta,
\qquad 0\le\tau\le1.
\end{equation*}
This path remains admissible. The nominal polynomial is Hurwitz by
$bk_pk_d>k_i$, and the preceding argument excludes any imaginary-axis
crossing along the path.  Since the characteristic polynomial is monic of fixed degree and
depends continuously on the path parameter, its roots cannot leave
the open left half-plane without such a crossing. Therefore every admissible linear closed loop is
Hurwitz.
\end{proof}

\subsection{PID endpoint coincidence}

% 原来 Theorem~\ref{thm:pid-region-coincidence} 后面的完整证明
\begin{proof}[Proof of Theorem \ref{thm:pid-region-coincidence}]
In view of Theorem~\ref{thm:main} and the inclusions
\(
\Omega_{\rm pid}^{\rm cq}
\subseteq
\Omega_{\rm pid}
\subseteq
\Omega_{\rm pid}^{\rm lin},
\)
it suffices to prove
\[
\Omega_{\rm pid}^{\rm cq}
=
\Omega_{\rm pid}^{\rm lin}.
\]

For $t>0$, corresponding to $t=\omega^2$, define
\begin{equation*}
S(t):=
\sqrt{(k_i-k_dt)^2+k_p^2t},
\qquad
R(t):=L_1\sqrt t+L_2t,
\end{equation*}
and
\begin{equation*}
F(t):=
bS(t)^2-k_pt^2-R(t)S(t).
\end{equation*}
Thus the frequency inequality in
\eqref{eq:pid-linear-region} is equivalent to
\begin{equation}
\label{eq:F-positive}
F(t)>0,\qquad \forall\,t>0.
\end{equation}

\emph{Step 1: $
\Omega_{\rm pid}^{\rm cq}
\subseteq
\Omega_{\rm pid}^{\rm lin}$.}
Suppose that
$(k_p,k_i,k_d)\in\Omega_{\rm pid}^{\rm cq}$.
By Theorem~\ref{thm:main}, there exist $\mu,\rho\in(0,1)$ such that
\[
\Pi_{\mu,\rho}>0,\qquad
\Delta_{\mu,\rho}>0,\qquad
k_i<
\frac{\Pi_{\mu,\rho}}
{2\bigl(k_d-\sqrt{\Delta_{\mu,\rho}}\bigr)}.
\]
Define
\begin{align*}
H(t):=
b(1-\mu)S(t)^2-k_pt^2
-\frac{L_1^2t}{4b\mu\rho}
-\frac{L_2^2t^2}{4b\mu(1-\rho)}.
\end{align*}
By the definitions of $\Pi_{\mu,\rho}$ and
$\Delta_{\mu,\rho}$,
\begin{equation}
\label{eq:H-polynomial}
H(t)
=
b(1-\mu)
\left[
\Delta_{\mu,\rho}t^2
+\bigl(\Pi_{\mu,\rho}-2k_ik_d\bigr)t
+k_i^2
\right].
\end{equation}
Moreover, direct completion of squares gives the exact identity
\begin{align}
F(t)-H(t)
&=
\left(
\sqrt{b\mu\rho}\,S(t)
-\frac{L_1\sqrt t}{2\sqrt{b\mu\rho}}
\right)^2
\nonumber\\
&+
\left(
\sqrt{b\mu(1-\rho)}\,S(t)
-\frac{L_2t}{2\sqrt{b\mu(1-\rho)}}
\right)^2
\ge0.
\label{eq:exact-young-identity}
\end{align}
The quadratic polynomial in \eqref{eq:H-polynomial} satisfies
\begin{align*}
&\Delta_{\mu,\rho}t^2
+\bigl(\Pi_{\mu,\rho}-2k_ik_d\bigr)t+k_i^2\\
&\quad=
\left(\sqrt{\Delta_{\mu,\rho}}\,t-k_i\right)^2
+
\left[
\Pi_{\mu,\rho}
-2k_i\bigl(k_d-\sqrt{\Delta_{\mu,\rho}}\bigr)
\right]t.
\end{align*}
The coefficient in square brackets is strictly positive. Hence
$H(t)>0$ for every $t>0$, and
\eqref{eq:exact-young-identity} yields \eqref{eq:F-positive}.

It remains to verify the Routh--Hurwitz condition. Set
$t_0:=k_i/k_d$. Since
$S(t_0)=k_p\sqrt{t_0}$, \eqref{eq:F-positive} gives
\begin{align*}
0<F(t_0)
&=
k_pt_0
\left(
bk_p-L_1-L_2\sqrt{t_0}-t_0
\right).
\end{align*}
Therefore
\[
bk_p>t_0=\frac{k_i}{k_d},
\]
and hence $bk_pk_d>k_i$. This proves
\(
\Omega_{\rm pid}^{\rm cq}
\subseteq
\Omega_{\rm pid}^{\rm lin}.
\)

\emph{Step 2: $
\Omega_{\rm pid}^{\rm lin}
\subseteq
\Omega_{\rm pid}^{\rm cq}$.}

Let $(k_p,k_i,k_d)\in \Omega_{\rm pid}^{\rm lin}$. Then $F(t)>0$ for all $t>0$. Define
\begin{equation}
\label{eq:sigma-frequency}
\sigma(t):=
\frac{
R(t)+\sqrt{R(t)^2+4bk_pt^2}
}{2b}.
\end{equation}
This is the positive root of
\[
bs^2-R(t)s-k_pt^2=0.
\]
Consequently,
\begin{equation}
\label{eq:F-sigma-equivalence}
F(t)>0
\quad\Longleftrightarrow\quad
S(t)>\sigma(t).
\end{equation}

Define
\[
Z(t):=\sigma(t)^2-k_p^2t.
\]
Evaluating the polynomial
$bs^2-R(t)s-k_pt^2$ at $s=k_p\sqrt t$ shows that
\begin{equation}
\label{eq:Z-sign}
Z(t)>0
\quad\Longleftrightarrow\quad
t+L_2\sqrt t+L_1>bk_p.
\end{equation}
Taking again $t_0=k_i/k_d$ in \eqref{eq:F-positive} gives
\[
t_0+L_2\sqrt{t_0}+L_1<bk_p.
\]
In particular, $bk_p>L_1$. Let $t_c>0$ be the unique solution of
\begin{equation}
\label{eq:tc-definition}
t_c+L_2\sqrt{t_c}+L_1=bk_p.
\end{equation}
The preceding inequality implies $t_0<t_c$.

For $t>t_c$, one has $Z(t)>0$ and $k_dt>k_i$. Hence
\eqref{eq:F-sigma-equivalence} is equivalent to
\begin{equation}
\label{eq:kd-frequency-bound}
k_d>
\psi(t):=
\frac{k_i+\sqrt{Z(t)}}{t}.
\end{equation}
For $0<t<t_c$, one has
$\sigma(t)<k_p\sqrt t\le S(t)$, so no additional restriction
arises. At $t=t_c$, strictness follows from $t_0<t_c$.

Thus, under $k_i/k_d<t_c$, \eqref{eq:F-positive} is equivalent to
$k_d>\psi(t)$ for every $t>t_c$. Define
\[
D_*:=\sup_{t>t_c}\psi(t).
\]
Before replacing these pointwise strict inequalities by $k_d>D_*$,
we must prove that this supremum is attained.

First we verify the asserted behavior at $t_c$. Set
\[
g(t):=t+L_2\sqrt t+L_1-bk_p,\qquad h(t):=k_p\sqrt t,
\qquad p(t,s):=bs^2-R(t)s-k_pt^2.
\]
Then $g(t_c)=0$, $g'(t_c)=1+L_2/(2\sqrt{t_c})>0$,
$p(t,\sigma(t))=0$, and
\[
p(t,h(t))=-k_pt\,g(t),\qquad
\partial_s p(t,\sigma(t))
=\sqrt{R(t)^2+4bk_pt^2}>0.
\]
Since $\sigma(t_c)=h(t_c)$, differentiating the two identities for
$p$ at $t_c$ and subtracting gives
\[
\sigma'(t_c)-h'(t_c)
=\frac{k_pt_cg'(t_c)}{\sqrt{R(t_c)^2+4bk_pt_c^2}}>0.
\]
Consequently,
\[
Z'(t_c)
=2k_p\sqrt{t_c}\bigl(\sigma'(t_c)-h'(t_c)\bigr)>0.
\]
Here $\partial_s p$ is the partial derivative in the scalar variable $s$.
Taylor expansion, with $O(\cdot)$ denoting a remainder bounded in
absolute value by a constant times its argument in the indicated limit,
now gives, as $t\downarrow t_c$,

\begin{equation*}
\psi(t)
=
\frac{k_i}{t_c}
+
\frac{\sqrt{Z'(t_c)}}{t_c}\sqrt{t-t_c}
+O(t-t_c).
\end{equation*}
Hence $\psi(t)>k_i/t_c$ immediately to the right of $t_c$.
On the other hand, as $t\to\infty$,
\begin{equation*}
\psi(t)
=
\kappa_\infty
+\frac{c_\infty}{\sqrt t}
+O(t^{-1}),
\end{equation*}
where
\begin{equation*}
\kappa_\infty
:=
\frac{L_2+\sqrt{L_2^2+4bk_p}}{2b},
\qquad
c_\infty
:=
\frac{L_1\kappa_\infty}
{\sqrt{L_2^2+4bk_p}}>0.
\end{equation*}
Let
\[
\ell:=\max\left\{\frac{k_i}{t_c},\,\kappa_\infty\right\}.
\]

If $\ell=k_i/t_c$, the first expansion supplies a finite point with
$\psi(\widehat t)>\ell$; if $\ell=\kappa_\infty$, the second does so
because $c_\infty>0$. Choose a real number $a$ strictly between
$\ell$ and $\psi(\widehat t)$. The two endpoint limits give
$t_c<t_-<\widehat t<t_+<\infty$ such that $\psi(t)<a$ for
$t\in(t_c,t_-)$ and for $t>t_+$. The maximum of $\psi$ on
$[t_-,t_+]$ therefore equals its supremum on $(t_c,\infty)$.
In particular, for some $t_*\in(t_c,\infty)$,
\[
D_*=\psi(t_*)>\max\{k_i/t_c,\kappa_\infty\}.
\]
Pointwise strictness at this maximizer now gives $k_d>D_*$ for the
original gain. Conversely, $k_d>D_*$ implies $k_i/k_d<t_c$ and
$k_d>\psi(t)$ for every $t>t_c$, and hence gives $F(t)>0$ for all
$t>0$. We have therefore justified the strict equivalence
\begin{equation}
\label{eq:kd-Dstar}
F(t)>0\ \text{for every }t>0
\quad\Longleftrightarrow\quad k_d>D_*.
\end{equation}

Consider the boundary value
\(
k_d^\partial:=D_*.
\)

To keep the original and boundary gains distinct, write
\[
S_\partial(t):=\sqrt{(k_i-k_d^\partial t)^2+k_p^2t},\qquad
F_\partial(t):=bS_\partial(t)^2-k_pt^2-R(t)S_\partial(t).
\]
Since $D_*>k_i/t_c$, one has $k_i/k_d^\partial<t_c$.
For $t>t_c$, the inequality $D_*\ge\psi(t)$ yields
$k_d^\partial t-k_i\ge\sqrt{Z(t)}$, and thus
$S_\partial(t)\ge\sigma(t)$. For $0<t<t_c$, strictness follows from
$\sigma(t)<k_p\sqrt t\le S_\partial(t)$; at $t=t_c$, it follows
from $k_d^\partial t_c>k_i$. Therefore
\[
F_\partial(t)\ge0\quad(t>0),\qquad F_\partial(t_*)=0,
\]
where equality at $t_*$ follows from $D_*=\psi(t_*)$.
This proves boundary feasibility without applying a strict-gain
criterion at equality.

Let
\[
S_*:=
\sqrt{(k_i-k_d^\partial t_*)^2+k_p^2t_*}
=\sigma(t_*).
\]
Define
\begin{equation}
\label{eq:optimal-mu-rho}
\mu_*:=
\frac{R(t_*)}{2bS_*},
\qquad
\rho_*:=
\frac{L_1\sqrt{t_*}}{R(t_*)}.
\end{equation}
Because $L_1,L_2>0$, one has $\rho_*\in(0,1)$. Furthermore, the
boundary equality
\(
bS_*^2-k_pt_*^2-R(t_*)S_*=0
\)
gives
\[
\mu_*
=
\frac12
\left(
1-\frac{k_pt_*^2}{bS_*^2}
\right)
\in\left(0,\frac12\right).
\]
Thus $(\mu_*,\rho_*)\in(0,1)^2$.

Use \eqref{eq:H-polynomial}, denoting its left-hand side by
$H_\partial(t)$, with
$k_d=k_d^\partial$ and
$(\mu,\rho)=(\mu_*,\rho_*)$. By
\eqref{eq:optimal-mu-rho}, both squares in
\eqref{eq:exact-young-identity} vanish at $t=t_*$. Hence
\[
H_\partial(t_*)=0.
\]
Since $F_\partial(t)\ge0$ and $F_\partial(t_*)=0$ at the interior point $t_*$,
one also has $F_\partial'(t_*)=0$. The derivatives of the two squares in
\eqref{eq:exact-young-identity} vanish at their common zero, and
therefore
\[
H_\partial'(t_*)=0.
\]

All functions in the square terms are differentiable near $t_*>t_c>0$;
the derivative of a square is zero at a zero of its underlying function.
Thus both the value and derivative of $H_\partial$ vanish at $t_*$.
Its constant term is $b(1-\mu_*)k_i^2>0$, so it is not identically zero
and cannot be a polynomial of degree at most one with this double zero.
Consequently,
\[
H_\partial(t)=b(1-\mu_*)\frac{k_i^2}{t_*^2}(t-t_*)^2.
\]
Comparing coefficients proves the identities below. Their strict
positivity also uses
$k_d^\partial-k_i/t_*=\sqrt{Z(t_*)}/t_*>0$.

We obtain
\begin{equation}
\label{eq:boundary-Pi-Delta}
\Delta_*=\frac{k_i^2}{t_*^2}>0,
\qquad
\Pi_{\mu_*,\rho_*}
=
2k_i
\left(
k_d^\partial-\frac{k_i}{t_*}
\right)>0,
\end{equation}
where $\Delta_*$ denotes
$\Delta_{\mu_*,\rho_*}$ evaluated at
$k_d=k_d^\partial$. In particular,
\[
\sqrt{\Delta_*}=\frac{k_i}{t_*}
\]
and
\begin{equation}
\label{eq:boundary-integral-equality}
k_i
=
\frac{\Pi_{\mu_*,\rho_*}}
{2\bigl(k_d^\partial-\sqrt{\Delta_*}\bigr)}.
\end{equation}

Return now to the original gain $k_d>D_*$ and keep
$(\mu_*,\rho_*)$ fixed. Put
\begin{equation*}
A_*:=
\frac{k_p}{b(1-\mu_*)}
+
\frac{L_2^2}
{4b^2\mu_*(1-\rho_*)(1-\mu_*)}.
\end{equation*}
Then
\(
\Delta_{\mu_*,\rho_*}=k_d^2-A_*,
\
\Delta_*=(k_d^\partial)^2-A_*.
\)
Since the function
\[
r\longmapsto r-\sqrt{r^2-A_*}
\]
is strictly decreasing for $r>\sqrt{A_*}$, it follows that
\begin{align*}
k_d-\sqrt{\Delta_{\mu_*,\rho_*}}
<
k_d^\partial-\sqrt{\Delta_*}
=
\frac{\Pi_{\mu_*,\rho_*}}{2k_i}.
\end{align*}

Indeed, $A_*>0$ and
$\frac{d}{dr}(r-\sqrt{r^2-A_*})=1-r/\sqrt{r^2-A_*}<0$.
Moreover, $\Delta_*>0$ implies $k_d^\partial>\sqrt{A_*}$, so the
whole interval from $k_d^\partial$ to $k_d$ lies in this domain.
For the fixed pair $(\mu_*,\rho_*)$, $\Pi_{\mu_*,\rho_*}$ does not
depend on $k_d$, while $\Delta_{\mu_*,\rho_*}=k_d^2-A_*$ increases
strictly and remains positive. The displayed denominators are positive
because $0<k_d^2-A_*<k_d^2$. Thus increasing $k_d$ from its boundary
value makes the integral-gain inequality strictly feasible without
changing the selected multipliers.

Consequently,
\[
k_i<
\frac{\Pi_{\mu_*,\rho_*}}
{2\bigl(k_d-\sqrt{\Delta_{\mu_*,\rho_*}}\bigr)}.
\]
Together with \eqref{eq:boundary-Pi-Delta}, this proves that
the gain satisfies the explicit conditions in
\eqref{eq:exact-cq-region}, and hence belongs to
$\Omega_{\rm pid}^{\rm cq}$ by Theorem~\ref{thm:main}.
Therefore,
\(
\Omega_{\rm pid}^{\rm lin}
\subseteq
\Omega_{\rm pid}^{\rm cq}.
\)
The two inclusions prove
\[
\Omega_{\rm pid}^{\rm cq}
=
\Omega_{\rm pid}^{\rm lin}.
\]
Combining this identity with Theorem~\ref{thm:main} and the
common-quadratic/system-level/linear inclusions gives
\eqref{eq:pid-all-regions-equal}.
\end{proof}

\subsection{Exact PI characterization}

\begin{proof}[Proof of Theorem \ref{thm:pi-exact-uniform}]
We prove $(iii)\Rightarrow(i)\Rightarrow(ii)\Rightarrow(iii)$.

\emph{Step 1: $(iii)\Rightarrow(i)$.}
The sufficiency argument is valid, in fact, for every $m\ge1$.
Condition~(iii) implies $k_p>0$ and $k_i>0$. For a given reference $y^*$, let $u^*$ satisfy $f(y^*,u^*)=0$, and define
\[
z_0:=\eta-\frac{u^*}{k_i},
\quad
z_1:=e,
\quad
q:=k_i z_0+k_p z_1.
\]
The mean-value theorem and the Jacobian bounds defining
$\mathcal F_{L,b}$ yield
\(
f(y^*-z_1,u^*+q)
=
-D(z)z_1+\Theta(z)q,
\)
for some state-dependent matrices satisfying
\[
\norm{D(z)}\le L,
\quad
\Sym(\Theta(z))\succeq bI_m.
\]
Hence, with $z=\operatorname{col}(z_0,z_1)$,
\[
\dot z
=
\mathcal A_{\rm PI}\bigl(D(z),\Theta(z)\bigr)z.
\]

Put
\[
\gamma:=\frac{k_i}{k_p},
\quad
\delta:=bk_p-\gamma-L.
\]
The inequality in~(iii) is equivalent to $\delta>0$. Consider
\begin{equation}
\label{eq:pi-P0}
P_0:=
\begin{bmatrix}
\gamma(2bk_p-L)&\gamma\\
\gamma&1
\end{bmatrix}.
\end{equation}
Since
\[
\det P_0
=
\gamma(2bk_p-L-\gamma)
=
\gamma(bk_p+\delta)>0,
\]
we have $P_0\succ0$.

Let $P=P_0\otimes I_m$ and define
\[
Q(D,\Theta):=
-\left(
P\mathcal A_{\rm PI}(D,\Theta)
+\mathcal A_{\rm PI}(D,\Theta)^\top P
\right).
\]
Writing $H=\Sym(\Theta)$, direct calculation gives
\[
Q(D,\Theta)
=
Q_b(D)
+
2k_p
\begin{bmatrix}
\gamma^2&\gamma\\
\gamma&1
\end{bmatrix}
\otimes(H-bI_m),
\]
where
\[
Q_b(D)=
\begin{bmatrix}
2bk_p\gamma^2I_m
&
\gamma(LI_m-D)\\
\gamma(LI_m-D^\top)
&
2\bigl((bk_p-\gamma)I_m-\Sym(D)\bigr)
\end{bmatrix}.
\]
Thus $Q(D,\Theta)\succeq Q_b(D)$. The Schur complement of the
upper-left block of $Q_b(D)$ is
\[
\begin{aligned}
S(D)
={}&2\bigl((bk_p-\gamma)I_m-\Sym(D)\bigr)\\
&-\frac{1}{2bk_p}
(LI_m-D^\top)(LI_m-D).
\end{aligned}
\]
For any nonzero $v$, let
\[
r_v:=\frac{v^\top\Sym(D)v}{\norm v^2}.
\]
Since $\norm D\le L$, we have
\(
\norm{(LI_m-D)v}^2
\le
2L(L-r_v)\norm v^2.
\)
Consequently,
\[
\frac{v^\top S(D)v}{\norm v^2}
\ge
2\delta+
\left(2-\frac{L}{bk_p}\right)(L-r_v)
\ge2\delta>0.
\]
Indeed, $\delta>0$ implies $bk_p>L$, while $r_v\le L$. Hence
$Q_b(D)\succ0$ for every $\norm D\le L$. By compactness of the matrix
ball $\{D:\norm D\le L\}$, there exists $\alpha>0$, depending only on
$(k_p,k_i,L,b)$, such that
\[
Q(D,\Theta)\succeq\alpha I_{2m}
\]
for all admissible $D$ and $\Theta$.

The standard Lyapunov estimate, together with equivalence of norms on
$\mathbb R^{2m}$, yields \eqref{eq:pi-uniform-estimate} with constants
depending only on $(k_p,k_i,L,b)$. In fact, this estimate holds for
arbitrary $z(0)$ and hence, in particular, under the prescribed
initialization $\eta(0)=0$. The same estimate precludes finite escape and
therefore proves forward completeness. Thus~(i) holds. Moreover,
$P_0\otimes I_m$ is a common quadratic certificate.

\emph{Step 2: $(i)\Rightarrow(ii)$.}
This follows immediately from
\[
\mathcal F_{\partial}^{m}\subset\mathcal F_{L,b}.
\]

\emph{Step 3: $(ii)\Rightarrow(iii)$.}
We first prove that the gains must be positive. Since
$0\in\mathbb N_0$, the boundary family contains
\[
f_0(x,u)=Lx+bI_m u.
\]
Choose $y^*=0$, so that $u^*=0$. In the unshifted augmented coordinates
$z=\operatorname{col}(\eta,e)$, the corresponding closed-loop matrix is
\[
\mathcal C_0
=
\begin{bmatrix}
0&I_m\\
-bk_iI_m&(L-bk_p)I_m
\end{bmatrix}.
\]

Statement~(ii), applied to arbitrary plant initial conditions with
$\eta(0)=0$, gives exponential decay for every initial state in
\[
\mathcal S:=\{0\}\times\mathbb R^m.
\]
As in the corresponding PID argument, this restricted set of initial
conditions is sufficient to test Hurwitz stability. Indeed,
\[
\mathcal S+\mathcal C_0\mathcal S=\mathbb R^{2m},
\]
and $e^{\mathcal C_0t}\mathcal C_0
=\mathcal C_0e^{\mathcal C_0t}$. Hence exponential decay on
$\mathcal S$ extends to $\mathcal C_0\mathcal S$ and therefore to all of
$\mathbb R^{2m}$. Thus $\mathcal C_0$ is Hurwitz.

Its scalar characteristic polynomial is
\[
p_0(s)
=
s^2+(bk_p-L)s+bk_i.
\]
The Routh--Hurwitz criterion now gives
\[
bk_p-L>0,
\qquad
bk_i>0.
\]
Since $b>0$ and $L\ge0$, we conclude that
\(
k_p>0,
k_i>0.
\)

It remains to prove the strict gain inequality. Put
\[
N:=bk_p^2-Lk_p-k_i.
\]
For $f_\nu\in\mathcal F_{\partial}^{m}$, the shifted closed-loop matrix is
\[
\mathcal C_\nu=
\begin{bmatrix}
0&I_m\\
-k_iB_\nu&LI_m-k_pB_\nu
\end{bmatrix},
\qquad
B_\nu=bI_m+\nu J_m.
\]
Consider the complexified invariant subspace associated with the eigenvalue
$\mathrm i$ of the leading $J_2$ block. On this subspace, the
characteristic equation is
\[
p_\nu(s)
=
s^2+\bigl(k_p(b+\mathrm i\nu)-L\bigr)s
+k_i(b+\mathrm i\nu)=0.
\]
Because $k_p>0$ and $k_i>0$ have already been established, its two roots
satisfy, as $\nu\to\infty$,
\begin{align*}
s_{\nu,1}
=
-\frac{k_i}{k_p}+O(\nu^{-1}),\quad
s_{\nu,2}
=
-\mathrm i k_p\nu-\frac{N}{k_p}+O(\nu^{-1}).
\end{align*}
In particular,
\(
\lim_{\nu\to\infty}\Re s_{\nu,2}
=
-\frac{N}{k_p}.
\)

To account for the prescribed initialization $\eta(0)=0$, choose a plant
initial condition supported on the first two coordinates. Identify this
real two-dimensional block with $\mathbb C$. Under this identification,
$z_1(0)=1$ represents the admissible real initial vector $(1,0)^\top$.
Normalize the initial state by
\(
z_0(0)=0,
z_1(0)=1.
\)
The corresponding error component is
\(
z_{1,\nu}(t)
=
a_{\nu,1}e^{s_{\nu,1}t}
+
a_{\nu,2}e^{s_{\nu,2}t},
\)
where
\begin{align*}
a_{\nu,1}
&=
\frac{s_{\nu,1}}{s_{\nu,1}-s_{\nu,2}}
=
O(\nu^{-1}),\\
a_{\nu,2}
&=
-\frac{s_{\nu,2}}{s_{\nu,1}-s_{\nu,2}}
=
1+O(\nu^{-1}).
\end{align*}
Thus the second mode is excited by an admissible plant initial condition
with a coefficient bounded away from zero.

If $N<0$, then $\Re s_{\nu,2}>0$ for all sufficiently large $\nu$, which
contradicts stability. If $N=0$, then $\Re s_{\nu,2}\to0$. Let $M$ and
$\lambda$ be the uniform constants in~(ii). Choose $\nu$ sufficiently large
that
\[
\Re s_{\nu,2}>-\frac{\lambda}{2},
\qquad
\Re s_{\nu,1}<\Re s_{\nu,2},
\]
and $|a_{\nu,2}|$ is bounded away from zero. Since the first mode then
decays strictly faster than the second, the reverse triangle inequality
gives, for all sufficiently large $t$,
\[
|z_{1,\nu}(t)|
\ge
\frac14e^{\Re s_{\nu,2}t}
>
Me^{-\lambda t},
\]
contradicting the uniform estimate in~(ii). Therefore $N>0$, namely,
\(
k_i+Lk_p<bk_p^2.
\)
Together with $k_p>0$ and $k_i>0$, this proves statement~(iii) and
completes the proof.
\end{proof}

\section{Conclusion}

Characterizing all PID and PI gains that guarantee uniform exponential regulation of uncertain nonlinear systems is a fundamental yet highly challenging problem in control theory.
This paper established exact system-level PID and PI gain regions for
two classes of non-affine MIMO systems with $m\geq2$. For the
second-order PID class, independently characterized common-quadratic
inner and linear-subclass outer regions were shown to coincide. Thus,
common-quadratic certification is lossless at the gain-region level,
and the linear subclass is a lossless outer test. For the first-order
PI class, uniform exponential regulation of a countable boundary-linear
family was shown to be equivalent to regulation of the full nonlinear
class, and the resulting exact condition admits a common quadratic
certificate. The PID region strictly enlarges the corresponding
existing sufficient region and yields a sequential tuning rule, while
the PI condition relaxes the corresponding existing bound by
$L^2/(4b)$.

Beyond these two gain-region formulas, the results raise a broader
structural question: for which nonlinear uncertainty classes can a
tractable Lyapunov-based inner region and a tractable test-family outer
region both coincide with the full system-level gain region? The
compact linear counterexample shows that common-quadratic losslessness
does not hold in general. Identifying verifiable structural conditions
that guarantee such an exact inner--outer coincidence is an important
direction for future research.

\end{document}